%% file: main.tex
\documentclass[11pt,a4paper]{article}

\usepackage[margin=1in]{geometry}
\usepackage{amsmath,amssymb,amsthm,mathtools}
\usepackage{array,tabularx,booktabs}
\usepackage[table]{xcolor}
\usepackage[most]{tcolorbox}
\usepackage[colorlinks=true,linkcolor=blue!55!black,citecolor=blue!55!black,urlcolor=blue!55!black]{hyperref}
\usepackage{graphicx}
\usepackage{subcaption}

\newtheorem{theorem}{Theorem}[section]
\newtheorem{lemma}[theorem]{Lemma}
\newtheorem{proposition}[theorem]{Proposition}
\newtheorem{corollary}[theorem]{Corollary}

\theoremstyle{definition}
\newtheorem{definition}[theorem]{Definition}

\theoremstyle{remark}
\newtheorem{remark}[theorem]{Remark}
\newtheorem{example}[theorem]{Example}

\allowdisplaybreaks

\newtcolorbox{protocolbox}[1][]{
  enhanced, breakable,
  colback=black!4, colframe=black!55, coltitle=white,
  fonttitle=\bfseries, title={#1},
  boxrule=0.6pt, arc=1.5pt, left=7pt, right=7pt, top=5pt, bottom=5pt,
  before skip=10pt, after skip=10pt
}
\definecolor{headgrey}{gray}{0.88}

\title{Device-Independent Conference Keys \\ from Parity-Extended Games}
\date{}
\author{Suvradip Chakraborty\thanks{Visa Research. \texttt{suvchakr@visa.com}}
\and
Ronak Ramachandran\thanks{UT Austin. \texttt{ronakr@utexas.edu} (Work done while at Visa.)}
\and
Aniruddha Sen \thanks{UT Austin. \texttt{aniruddhasen@utexas.edu} (Work done while at Visa.)}}

\begin{document}
\input{macros}

\maketitle

\input{abstract}

\pagenumbering{arabic}

\input{intro_new}

\input{technical-overview}
\input{prelims}
\input{parity_MSG1}

\input{CKA}
\input{security}

\input{noisy-key-rates}

\input{acknowledgements}

\bibliographystyle{alpha}
\bibliography{references}

\newpage
\appendix

\input{appendix}

\input{disclaimer}

\end{document}

%% file: macros.tex
\newcommand{\ket}[1]{\left|#1\right\rangle}
\newcommand{\bra}[1]{\left\langle#1\right|}
\newcommand{\braket}[2]{\left\langle#1\middle|#2\right\rangle}
\newcommand{\proj}[1]{\ket{#1}\!\bra{#1}}
\newcommand{\tr}{\operatorname{Tr}}
\newcommand{\id}{\mathbb{I}}
\newcommand{\GHZ}{\mathrm{GHZ}}
\newcommand{\Bell}[1]{\ket{\Phi^{#1}}}
\newcommand{\eps}{\varepsilon}
\newcommand{\wt}{\operatorname{wt}}
\newcommand{\pMSG}{\mathsf{pMSG}}
\newcommand{\MSG}{\mathsf{MSG}}
\newcommand{\pCHSH}{\mathsf{pCHSH}}
\newcommand{\Win}{\mathsf{W}}
\newcommand{\sgn}{\sigma}
\newcommand{\Hmin}{H_{\min}}
\newcommand{\leakEC}{\mathrm{leak}_{\mathrm{EC}}}
\newcommand{\Eve}{\mathsf{E}}
\newcommand{\cH}{\mathcal{H}}
\newcommand{\cM}{\mathcal{M}}
\newcommand{\cX}{\mathcal{X}}
\newcommand{\cA}{\mathcal{A}}
\newcommand{\wc}{\omega_{\mathrm C}}
\newcommand{\wq}{\omega_{\mathrm Q}}
\newcommand{\pguess}{p_{\mathrm{guess}}}
\newcommand{\Acore}{A_{\mathrm{core}}}

\newcommand{\as}[1]{{\color{teal}[AS: #1]}}
\newcommand{\rr}[1]{{\color{magenta}[RR: #1]}}

\newcommand{\Ftwo}{\mathbb{F}_2}
\newcommand{\pG}[1]{\mathsf{p}#1}
\newcommand{\Xsupp}{\operatorname{supp}_X}
\providecommand{\cY}{\mathcal{Y}}
\providecommand{\cB}{\mathcal{B}}

%% file: abstract.tex
\begin{abstract}
Device-independent conference key agreement (DI-CKA) lets a group of parties establish a shared secret key from untrusted quantum devices, with security certified by non-locality. Existing DI-CKA protocols are each built around a single Bell inequality, typically a multiparty variant of the CHSH game. DI-QKD protocols, in contrast, have been built from a much richer landscape of
non-local games, and it has remained unclear how to carry this landscape over to the conference setting. We introduce \emph{Parity-$G$ games}, which extend any two-player
game $G$ to $N$ players, for every $N$, provided $G$ has an optimal strategy in which one player measures Pauli observables. The extension preserves the quantum and classical values of $G$, and the security of the resulting $N$-party protocol follows from an analysis of the two-player game alone. Our framework recovers the Parity-CHSH game of Ribeiro, Murta and Wehner (Phys. Rev. A, 2018) as a special case. Applied to the Mermin--Peres Magic Square Game,
it yields a new $N$-player pseudo-telepathy game, the \emph{Parity Magic Square Game}, which ideal devices win in every round. We use it to construct the \emph{first} DI-CKA protocol based on a pseudo-telepathy game. We prove the protocol secure against coherent attacks. It produces up to two key bits per round, and at low noise its key rate exceeds that of the DI-CKA protocol based on the Parity-CHSH game.
\end{abstract}

%% file: intro_new.tex
\section{Introduction}\label{sec:intro_new}
Quantum key distribution (QKD) lets two parties establish a secret key whose security rests on the laws of quantum mechanics \cite{Bennett1984BB84,ekert1991quantum}. Standard security proofs assume that the devices behave as modeled, for instance that they prepare and measure qubits in prescribed bases. Real devices deviate from such models, and practical attacks (e.g., \cite{lydersen2010hacking})
exploit these deviations. Device-independent (DI)
QKD, also called fully DI QKD, removes the need to trust the quantum devices and seeks security even when those devices are uncharacterized or supplied by an adversary.\footnote{Fully DI QKD trusts neither the sources nor the measurement devices.
Measurement-device-independent QKD \cite{lo2012measurement,braunstein2012side} trusts the sources
but not the measurements, and one-sided DI QKD \cite{branciard2012one} trusts one party's
devices.} The parties treat their devices as black
boxes, and security follows from the observed input-output statistics alone, through the
winning probability of a non-local game \cite{acin2007device,arnon2018practical}. Vazirani and
Vidick proved that a slightly modified version of Ekert's protocol, which is based on the CHSH
game \cite{clauser1969proposed,ekert1991quantum}, is secure in this fully device-independent
setting \cite{vazirani2012fully}. DI-QKD has since been built on several other games, including
generalized CHSH inequalities \cite{sekatski2021device}, the Mermin--Peres magic square game
\cite{zhen2023device}, and, more broadly, on the class of monogamy-of-entanglement games \cite{cervero2025device}. The choice of
game determines what the test certifies and how much key each round yields.

Conference key agreement (CKA) generalizes QKD to $N$ parties who wish to establish a common secret key. When the parties share a multipartite entangled state such as the GHZ state, a single
measurement round gives all of them the same key bit. In some network topologies, this yields
higher key rates than running $N-1$ two-player QKD protocols \cite{epping2017multi,murta2020quantum}. DI-CKA has received much less attention than DI-QKD.
Ribeiro, Murta and Wehner \cite{ribeiro2018fully} gave the first fully DI-CKA protocol, based on the Parity-CHSH game,
an $N$-party extension of the CHSH game \cite{clauser1969proposed}. Subsequent work designed a genuinely multipartite Bell inequality for DI-CKA
\cite{holz2020genuine} and derived entropy bounds and key rates for multiparty DI protocols
\cite{grasselli2021entropy,grasselli2023boosting}. Wooltorton et~al.\ recently showed that DI-CKA
does not require genuine multipartite entanglement \cite{wooltorton2025genuine}.

Each of these protocols relies on a multiparty generalization of the CHSH inequality, such as the
Parity-CHSH inequality \cite{ribeiro2018fully} or the genuinely multipartite generalization of CHSH
by Holz et~al.\ \cite{holz2020genuine}. To our knowledge, DI-CKA protocols based on other
non-local games are not known. In particular, none is based on a \textit{pseudo-telepathy} game, that is, a game quantum players can win with certainty while classical players cannot \cite{brassard2005quantum}. Such games seem attractive for DI cryptography, since malicious devices cannot afford to lose even a single game in the ideal noiseless setting.
We have concrete evidence of such an advantage in the two-party setting: DI-QKD based on the magic square game, a pseudo-telepathy game, indeed yields more key per round than CHSH-based DI-QKD at low noise \cite{zhen2023device}. This prompts the following question:

\begin{quote}
\textbf{Question 1:} \itshape Would a DI-CKA protocol based on a pseudo-telepathy game outperform existing protocols?
\end{quote}

There are many obvious first choices for such an $N$-party pseudo-telepathy game: the $N$-player GHZ game or parallel two-player magic square games\footnote{See Page \pageref{par:parallel} for a discussion of DI-CKA protocols based on performing pairwise DI-QKD protocols in parallel.}, for instance.
The central difficulty, however, is proving security for such a protocol by upper-bounding the information that an adversary Eve can learn about the conference key the players are attempting to share. Unfortunately, the only known techniques for doing so require solving computationally intensive semi-definite programs \cite{navascues2007bounding,brown2024device} that become infeasible for $N$ even slightly larger than $2$ and do not say anything about asymptotic behavior as $N$ increases. Holz, Kampermann, and Bruß \cite{holz2020genuine}, for instance, are unable to prove any security guarantees on their $N$-party game or proposed DI-CKA protocol for more than $4$ parties. This leads to the following question:

\begin{quote}
\textbf{Question 2:} \itshape Is there a general way to construct $N$-party DI-CKA protocols with provable security for all $N$?
\end{quote}

\subsection{Our contributions}\label{sec:contributions}
We answer both questions affirmatively. Our main contributions are the following.
 
\begin{itemize}
    \item \textbf{Parity-$G$ games.} For every two-player game $G$ with an optimal strategy in which one player measures Pauli observables on shared EPR pairs (Definition~\ref{def:pauli-strategy}), we define an $N$-player game $\pG{G}_N$, the \emph{parity extension} of $G$ (Definition~\ref{def:parityG}). For every $N\ge2$, $\pG{G}_N$ has the same quantum and classical values as $G$ (Theorem~\ref{thm:parityG-values}), and its optimal value can be attained given copies of the $N$-party GHZ state (Proposition~\ref{prop:parityG-honest}).

    \item \textbf{DI-CKA protocols based on Parity-$G$ games.} We instantiate our Parity-$G$ games framework with the Mermin--Peres magic square game in particular, obtaining the \emph{Parity Magic Square game} $\pMSG_N$ (Definition~\ref{def:pmsg}), which is a new $N$-player pseudo-telepathy game with quantum value $1$ and classical value $8/9$ for every $N$ (Propositions~\ref{prop:wq}  and~\ref{prop:wc}). From it we construct the first DI-CKA protocol based on such a game, which produces up to two key bits per round, and we prove its security against coherent attacks (Sections~\ref{sec:protocol} and~\ref{sec:security-proofs}). Applied to the CHSH game, our Parity-$G$ framework also recovers the Parity-CHSH game of Ribeiro, Murta and Wehner \cite{ribeiro2018fully} as a special case (Example~\ref{ex:pchsh}), capturing the DI-CKA protocol based on it, and extends to other two-party games as well.

    \item \textbf{Reduction of security to the two-party game.} A major benefit of our framework is that the security of the $N$-party protocol reduces to that of the underlying two-party game $G$: every convex lower bound on the entropy of Alice's key output that holds for $G$ also holds for $\pG{G}_N$, for every $N$ (Lemma~\ref{lem:pmsg-Nto2-reduction}). This provides us with a family of DI-CKA protocols, based on Parity-$G$ games, in which the security analysis is relatively simple and does not increase in difficulty with the number of parties. This gives us an important advantage over the current best techniques for analyzing the security of general DI-CKA, which require solving much larger semi-definite programs via the NPA hierarchy \cite{navascues2007bounding,brown2024device} that quickly become infeasible for $N$ even slightly larger than $2$, and do not give any asymptotic guarantees.

\item \textbf{Key rates.} We compare the asymptotic key rates of the DI-CKA protocols based on $\pMSG_N$ and on the Parity-CHSH game under depolarizing noise (Section~\ref{sec:key-rates}). Under global depolarizing noise, our protocol based on $\pMSG_N$ achieves a higher key rate in the low
noise regime. Conversely, to show the tightness of our analysis, we also give an attack, adapted from \cite{adamson2020quantum}, under which an adversary learns Alice's entire key while the devices are able to win $\pMSG_N$ with probability up to $\tfrac59+\tfrac49\cos^2(\pi/8)\approx0.935$ (Appendix~\ref{app:msg_upper}). When this attack is combined with
\cite{lukanowski2023upper}, it yields an upper bound on the key rate. We also consider a local depolarizing noise model, where an advantage again persists for $\pMSG_N$ at low noise, though the key rates of both protocols decrease as $N$ increases.
\end{itemize}

%% file: technical-overview.tex
\section{Technical Overview}\label{sec:tech-overview}

Our technical overview is divided into three parts: an intuitive explanation of parity-extensions of 2-player games, an overview of our device-independent protocol for quantum conference key agreement based on the parity-extension of the Mermin--Peres Magic Square Game, and finally an analysis of key rates for this new protocol.

\subsection{$N$-Player Parity-Extensions of $2$-Player Games}

Optimal quantum strategies for 2-player non-local games often require players to share the state $$\Bell{+} := \frac{1}{\sqrt{2}}(\ket{00} + \ket{11}).$$ 
In many cases, however, replacing that state with 
$$\Bell{-} := \frac{1}{\sqrt{2}}(\ket{00} - \ket{11}) = (Z \otimes I) \Bell{+}$$ 
without Alice and Bob knowing has a very predictable effect on player outputs. 
In particular, because $\Bell{-}$ is just $\Bell{+}$ with a $Z$ gate locally applied to Alice's qubit,
if all of Alice's possible measurements commute or anticommute with $Z$---for instance, if her measurements are Paulis\footnote{Paulis are a strict subset; every basis along the equator of the Bloch sphere anticommutes with $Z$. Throughout this paper, we only restrict to Paulis for simplicity, but the results hold more generally.}---then performing her measurements on $\Bell{+}$ is equivalent to performing them on $\Bell{-}$ and flipping the outcome if the measurement anticommutes with $Z$. 

We call a $2$-player non-local game a \textit{Pauli game} if it has an optimal quantum strategy where parties share copies of $\Bell{+}$ and one party applies only Pauli measurements to their shares. 
In Section \ref{sec:game} we show that any Pauli game $G$ can be extended to an $N$-player non-local game Parity-$G$ that inherits the classical/quantum winning probabilities and entropy guarantees of the 2-player game. Almost every known 2-player non-local game is a Pauli game, so this technique is very general. Parity-CHSH is one special case.

The key insight enabling $N$-player parity-extensions of Pauli games is that if all but $2$ of $N$ players measure their qubit of a shared $N$-party GHZ state in the $X$ basis, the remaining two players---say Alice and Bob---will share $\Bell{+}$ or $\Bell{-}$, and which state they share will depend sensitively on every measurement outcome of the $N-2$ others. 
In particular, Alice and Bob will share $\frac{1}{\sqrt{2}}(\ket{00} + (-1)^p\ket{11})$, where $p \in \{0,1\}$ is the \textit{parity} of the $N-2$ measurement outcomes.
Alice and Bob can then use their shared state to play any Pauli game, so long as the referee flips Alice's raw output appropriately based on $p$ and Alice's input (which determines whether her measurement will anticommute with $Z$) before evaluating the game's winning condition.

Notably, the players other than Alice and Bob receive no input\footnote{Thus, unlike other $N$-party non-local games that require non-signaling between all $N^2$ pairs of players, Parity-$G$ games only require non-signaling between $N-1$ pairs (Alice and everyone else), which is quadratically easier to enforce.}. In fact, the new $N$-party game can be won by simply having Alice and Bob independently play the 2-player game using shared copies of $\Bell{+}$, while all other players output $0$, so this game is not self-testing for an $N$-qubit GHZ state. Consistent with \cite{wooltorton2025genuine}, our game also does not require genuine multi-partite entanglement to achieve DI-CKA. As we'll see, all ability to arrive at shared bits of key comes from the tension between playing this game during test rounds and passing the error correction step of our protocol.

\paragraph{The Parity Magic Square Game} Our DI-CKA protocol relies on the $N$-player parity-extension, $\pMSG_N$, of the $2$-player Mermin--Peres Magic Square Game, $\MSG$. 

In $\MSG$, Alice and Bob claim to have some $3 \times 3$ grid of $0$'s and $1$'s where the rows have an even number of $1$'s and the columns have an odd number of $1$'s. To check this claim, the referee asks Alice for the $3$ entries in one row and Bob for the $3$ entries in one column. The players win if Alice's entries have an even number of $1$'s, Bob's have an odd number of $1$'s, and Alice and Bob agree on the entry where their row and column intersect. 

This game is a Pauli game, since an optimal quantum strategy exists where Alice and Bob share two copies of $\Bell{+}$ and Alice makes only Pauli measurements on her two qubits, as shown in Figure \ref{fig:peres tech overview}.
\begin{figure}[h]
    \centering
    {\renewcommand{\arraystretch}{1.3}
    \begin{tabularx}{0.5\textwidth}{>{\centering\arraybackslash}X|>{\centering\arraybackslash}X|>{\centering\arraybackslash}X|>{\centering\arraybackslash}X|}
        & \cellcolor[gray]{0.9}$a_0$ & \cellcolor[gray]{0.9}$a_1$ & \cellcolor[gray]{0.9}$a_2$ \\\hline
        \cellcolor[gray]{0.9}$x=0$ & $\id\otimes Z$ & $Z\otimes\id$ & $Z\otimes Z$\\\hline
        \cellcolor[gray]{0.9}$x=1$ & $X\otimes\id$ & $\id\otimes X$ & $X\otimes X$\\\hline
        \cellcolor[gray]{0.9}$x=2$ & $-X\otimes Z$ & $-Z\otimes X$ & $Y\otimes Y$\\\hline
    \end{tabularx}}
    \caption{Alice's Pauli measurements for outputs $a_0, a_1, a_2 \in \{0,1\}$ on input $x \in \{0,1,2\}$.}
    \label{fig:peres tech overview}
\end{figure}

In one optimal quantum strategy for $\pMSG_N$, the two $\Bell{+}$ states are replaced by two GHZ states. Let $p_0$ and $p_1$ be the parities of the $X$ basis measurement outcomes of the players other than Alice and Bob for the first and second GHZ state respectively. To modify the winning condition of $\MSG$ to work for $\pMSG_N$, Alice's raw outputs must be flipped if $p_0$ is $1$ and her measurement anticommutes with $Z$ on the first qubit, and again if $p_1$ is $1$ and her measurement anticommutes with $Z$ on the second qubit.
In other words, before checking the winning condition of $\MSG$, the referee must XOR the bits specified in Figure \ref{fig:swap mask tech overview} into each of Alice's three outputs $a_0, a_1, a_2 \in \{0,1\}$ depending on her input $x \in \{0,1,2\}$.
\begin{figure}[h]
    \centering
    {\renewcommand{\arraystretch}{1.3}
    \begin{tabularx}{0.5\textwidth}{>{\centering\arraybackslash}X|>{\centering\arraybackslash}X|>{\centering\arraybackslash}X|>{\centering\arraybackslash}X|}
    & \cellcolor[gray]{0.9}$a_0$ & \cellcolor[gray]{0.9}$a_1$ & \cellcolor[gray]{0.9}$a_2$ \\\hline
    \cellcolor[gray]{0.9}$x=0$ & $0$ & $0$ & $0$ \\\hline
    \cellcolor[gray]{0.9}$x=1$ & $p_0$ & $p_1$ & $p_0 \oplus p_1$ \\\hline
    \cellcolor[gray]{0.9}$x=2$ & $p_0$ & $p_1$ & $p_0 \oplus p_1$ \\\hline
    \end{tabularx}}
    \caption{The referee's corrections to Alice's outputs for $\pMSG_N$.}
    \label{fig:swap mask tech overview}
\end{figure}

We call this table of correction bits the \textit{swap mask} of $\pMSG_N$. By hard-coding these corrections into the winning condition, we will have fully specified the game.

\subsection{The Conference Key Agreement Protocol}
Now, we show how to construct a conference key agreement protocol based on $\pMSG_N$. This protocol extends to Parity-$G$ games as well, but for ease of exposition we focus on just analyzing the protocol with $\pMSG_N$. We give an overview here, and the protocol is more formally described in Section \ref{sec:protocol}.

Consider $N$ parties $A, B_1, B_2,...,B_{N-1}$ where the first \emph{two} parties play distinguished roles, and are referred to as Alice ($A$) and Bob ($B_1$). Since we are in the device independent setting, let us assume some (untrusted) multipartite quantum state $\rho$ is shared between all the players. Broadly, we divide the protocol into $m$ rounds, with each round being designated as a \emph{test} round with probability $\gamma$ and a \emph{key} round with probability $1-\gamma$. The players do the following.
\begin{itemize}
    \item In key rounds, all players ask their devices to measure their state in the $Z$ basis.
    \item In test rounds, the players $B_2,...,B_{N-1}$ ask their devices to measure their state in the $X$ basis. Alice and Bob ask their devices to perform measurements according to the optimal quantum strategy of $\MSG$. Combined, this is the optimal winning strategy for $\pMSG_N$, as described in the last section.
\end{itemize}
After all the rounds have been completed, the players first do a \emph{parameter estimation} step, in which they publicly communicate their test-round outcomes, which allows them to calculate their average winning probability over multiple rounds of $\pMSG_N$. If the winning probability $\omega$ is below some predetermined threshold $\omega_{\mathrm{tol}}$, they abort the protocol. If not, they carry out an \emph{error correction} and \emph{privacy amplification} step to verify they all received the same raw key bits in the key rounds, and to extract fully secret bits of key from their raw key bits. Clearly, if the devices behaved honestly and shared a tensor product of two GHZ states with the players, they would get bits of key. But why should we expect this protocol to be secure? Intuitively, the security of this protocol mainly comes from the observation that the devices are unable to distinguish key rounds and \emph{certain} test rounds. In particular, Alice's device can't distinguish between the case where it is a key round, or the case when her test round input for $\MSG$ is $0$ (which is the first row). Both of these cases correspond to Alice's device being asked to measure $Z \otimes I$ and $I \otimes Z$. Also note that the designation of \emph{each} round is probabilistically chosen. So, in order to pass the winning condition on the test rounds, the devices are forced to play honestly and report correct measurement outcomes. Thus, on key rounds they are forced to do the same as well. 

Note that in order to just win on the test rounds, the devices could just share a tensor product of two EPR pairs with Alice and Bob, and have the other players just output 0. This ensures Alice and Bob get bits of secret key, but the other players do not. However, this is noticed by the players in the error correction step and the protocol is aborted. Thus, the devices are forced to act honestly in order to pass both these checks.

\subsection{Security analysis and key rates}
There are two parameter choices the players have when carrying out the above conference key agreement protocol: $\omega_{\mathrm{tol}}$, the tolerated winning probability, and $\gamma$, the test-round probability. Let us start by considering the security of a particular round -- in particular we will just focus on setting $\omega_{\mathrm{tol}}$. The ideal winning probability of $\pMSG_N$ is 1, but we will often wish to set $\omega_{\mathrm{tol}} < 1$, since there may be noise in the quantum channels that brings the winning probability down. Fortunately, it turns out that we can still get bits of secret key depending on how far below 1 we set $\omega_{\mathrm{tol}}$. Note that we must still set $\omega_{\mathrm{tol}} > 8/9$ to get any non-zero bits of key, since it is known that classical adversaries (the devices) can win $\pMSG_N$ deterministically with winning probability $8/9$, so their winning strategy or reported measurement outcomes could be pre-decided by an adversary Eve. In fact, we show that we need $\omega_{\mathrm{tol}} \gtrsim 0.935$, since a better \emph{quantum} strategy exists for the devices, as we show in Appendix \ref{app:msg_upper}.

We formally prove the per round security of the above protocol by lower bounding the per round \emph{conditional von Neumann entropy} $H(A|E)$ of Alice's key bitstring, over all possible shared states and measurement strategies which can be devised by Eve (who we assume constructs the untrusted devices), such that she succeeds with probability $\ge \omega_{\mathrm{tol}}$ on $\pMSG_N$. Intuitively, this just lower bounds how many bits of secret key Alice is expected to get per round as a function of $\omega_{\mathrm{tol}}$, so for us $0 \le H(A|E) \le 2$. No known analytical solution for this exists, so we resort to using numerical techniques which involve solving semi-definite programs (SDPs) via the NPA hierarchy \cite{navascues2007bounding,brown2024device}, to calculate a lower bound. It is often challenging to calculate such lower bounds in device-independent protocols \cite{brown2024device}, especially if the number of players is large, since this makes the size of the SDPs blow up. However, we get around this problem for our protocol by showing a reduction of the security of the $N$-party $\pMSG_N$ to the two party $\MSG$, by Lemma \ref{lem:pmsg-Nto2-reduction}. The proof relies on expressing the entropy lower bound of $\pMSG_N$ as a convex combination of entropy lower bounds of all possible $\MSG$ games Alice and Bob could be playing (up to a relabeling of their inputs) in the test rounds. Significantly, this technique generalizes to Parity-$G$ games, which provides a new framework for proving security of a large family of conference key agreement protocols by reducing to the much simpler security bound of just 2 party games. Finally, we extend our security over all $n$ rounds against coherent attacks (the protocol has $n$ rounds and $N$ parties), by appealing to the Generalized Entropy Accumulation Theorem \cite{metger2022generalised}. We refer to Section \ref{sec:security-proofs} for our precise formulation and proofs.

\paragraph{Key rates} We evaluate our protocol by analyzing its \emph{asymptotic key rate}. The key rate $r_\infty = H(A|E)-H(A|B)$ quantifies the average number of key bits we get each round over all rounds, when setting the test round probability $\gamma \rightarrow 0$ as the number of rounds increases. It is a function of $\omega_{\mathrm{tol}}$, and accounts for both the entropy lower bound when $\omega < 1$, and the security loss in the error correction step. 

We mainly compare our conference key agreement protocol based on $\pMSG_N$ with the work of \cite{ribeiro2018fully}, whose protocol is based on the \emph{Parity-CHSH} game. As the name suggests, our framework of Parity-$G$ games captures the Parity-CHSH game, as explained in Example \ref{ex:pchsh}, roughly since it can be expressed as a suitable $N$-party extension of the CHSH game. Moreover, we show that our protocol based on $\pMSG_N$ achieves strictly better key rates in the low noise regime. Note that $r_\infty$ can be expressed as a function of noise since it influences the expected winning probability $\omega$ of our $N$-party games, and we set the threshold $\omega_{\mathrm{tol}}$ accordingly. 

For instance, our protocol achieves higher key rate when the \emph{global} depolarizing noise probability $q$ is below $\approx 5\%$; also when the \emph{local} depolarizing noise probability is below $\approx 1.5\%$, setting the number of parties $N=3$. Conversely, our protocol has lower key rate when $q \gtrsim 10\%$. Both protocols exhibit zero key rate when $q \gtrsim 14\%$.  We calculate both lower \emph{and} upper bounds on key rates for both protocols to show the gap in tightness of our comparison. We refer to Section \ref{sec:key-rates} for a detailed explanation of our techniques and comparison between key rates of both protocols.

%% file: prelims.tex
\section{Preliminaries}\label{sec:prelim}

\subsection{Quantum states and entropies}
We assume the reader is familiar with quantum information essentials. For a primer on quantum states and measurements, see \cite{nielsen2000quantum}. We briefly introduce some useful definitions and notation regarding states and entropies.

\begin{itemize}
    \item The canonical \emph{Bell state} or \emph{EPR pair} is defined as $\Bell{+}=\frac{1}{\sqrt{2}}(\ket{00} + \ket{11})$, the other Bell states are $\Bell{-}, \ket{\Psi^+}, \ket{\Psi^-}$.
    \item A \emph{classical-quantum} state is defined as a state $\rho_{KE} = \sum_{x \in \mathcal{X}}\ket{x}\bra{x} \otimes \rho_{E,x}$ where $E$ is a quantum system and $\mathcal{X}$ is a classical alphabet.
    \item The \emph{von Neumann entropy} of a state $\rho$ is defined as $H(A)_\rho=-\tr(\rho_A\log\rho_A)$. 
    \item The \emph{conditional von Neumann entropy} of a bipartite state $\rho_{AB}$ is defined as $H(A|B)_\rho=H(AB)_\rho-H(B)_\rho$. 
    \item The \emph{min-entropy} of a classical-quantum state $\rho_{KE}$  is defined as $\Hmin(K|E)_\rho=\\ -\log_2(\pguess(K|E))$, where $\pguess(K|E)=\max_{\{M_k\}}\sum_k\tr(M_k\rho_E^k)$ is the optimal probability of guessing the classical bitstring $K$ by measuring the state $\rho_E$. 
    \item The \emph{smooth min-entropy} of a classical-quantum state is defined as $\Hmin^\eps(K|E)_\rho = \\ \sup_{\rho'}{H_{\min}(K|E)}_{\rho'}$, where the supremum is over $\eps$-close states $\rho'$ in terms of their purified distance $\sqrt{1-F(\rho,\rho')^2} \le \eps$, where $F(\rho,\rho')$ denotes the fidelity between $\rho$ and $\rho'$. This is the variant of min-entropy we will mainly be using throughout the paper, since it governs randomness extraction.
    \item The \emph{min-tradeoff function} is a function which lower bounds the conditional von Neumann entropy of a particular non-local game, as a function of its winning probability.
\end{itemize}

\subsection{Non-Local Games and Pseudo-Telepathy}

Informally, non-local games are collaborative games where players that cannot communicate with each other receive inputs from a referee and must return satisfactory outputs to the referee.

\begin{definition}[Non-local game]\label{def:nlg}
An $N$-player non-local game $G=(\pi,\vec\cX,\vec\cA,V)$ has finite input sets $\cX_j$,
output sets $\cA_j$, an input distribution $\pi$ on $\vec\cX=\cX_1\times\cdots\times\cX_N$,
and a predicate $V:\vec\cX\times\vec\cA\to\{0,1\}$. A referee samples $\vec x\sim\pi$, sends
$x_j$ to player $j$, each player outputs
$a_j$ \emph{without communicating} with the others, and all players win together iff $V(\vec x,\vec a)=1$. 
\end{definition}

Preventing communication between players is crucial, since communicating players could trivially win any such game. In practice, non-communication is enforced by placing players far enough apart that signals, which are limited by the speed of light, cannot travel between players before the game ends. This is why we call such games \textit{non-local}---players are separated from each other. Since communication must be prevented between every pair of players, the difficulty of preventing communication scales as $N^2$, becoming more demanding as $N$ grows.

Every quantum strategy for a non-local game is defined by a shared quantum state and input-dependent observables (POVMs) that players can apply to their part of the state\footnote{In this paper we will restrict to local observables rather than more general commuting observables.}. 
The classical and quantum values of a non-local game $G$, denoted $\wc(G)$ and $\wq(G)$, are the maximum possible winning probabilities of classical and quantum strategies respectively. 
Because classical strategies are a strict subclass of quantum strategies where players can only share randomness instead of quantum states, $\wc(G)\le\wq(G)$. 

What makes non-local games interesting is that quantum players can win some games with \textit{strictly greater} probability than classical players by using entanglement to correlate their outputs in ways that classical players cannot. 
For this reason, winning such a game with much higher probability than $\wc(G)$ can certify that players shared some highly entangled state. Games which certify a specific state are called \textit{self-testing} \cite{mayers2004selftesting}. 

Games where the optimal quantum strategy succeeds with certainty, called \emph{pseudo-telepathy} games, can be especially powerful, since losing a single game certifies non-classicality.

\begin{definition}[Pseudo-telepathy]
$G$ is a \emph{pseudo-telepathy} game if $\wq(G)=1>\wc(G)$. 
\end{definition}

\subsection{The Mermin--Peres magic square game}\label{sec:msg}
Here, we introduce a well-known non-local game known as the Magic Square game or $\MSG$; it will be useful in our construction of a multiparty game later on. We also give brief proofs of its optimal classical and quantum winning strategies.

\begin{definition}[Magic Square game, $\MSG$]\label{def:msg}
This is a non-local game played between two players, Alice and Bob.
Alice and Bob both receive inputs $x,y$ respectively from the set $\{0,1,2\}$. We assume their inputs are uniform random and independent. They each output two vectors $a = (a_0,a_1)$ and $b = (b_0,b_1)$; then set $a_2 = a_0 \oplus a_1$, $b_2 = b_0 \oplus b_1 \oplus 1$. They win iff $a_y=b_x$.
\end{definition}
For intuition, it is useful to think of a strategy for this game as a $3 \times 3$ matrix $M$, where Alice and Bob's inputs are a row and column index of the matrix. They have to output an even parity row and odd parity column respectively, and they win iff their intersecting cell has the same value.

\begin{proposition}[Classical value of $\MSG$]\label{prop:msgc}
$\wc(\MSG)=8/9$.
\end{proposition}

\begin{proof}
Consider Alice and Bob's classical strategies as a matrix. By convexity, it suffices to restrict to deterministic strategies $x\mapsto a^{(x)}$, $y\mapsto b^{(y)}$. For a strategy which wins over all possible cells that can appear at the intersection of Alice and Bob's output, the matrix $M_{xy}:=a^{(x)}_y$ would need to have $a^{(x)}_y=b^{(y)}_x$ for
all $x,y$. However, summing all entries by rows (Alice even) gives $\bigoplus_{x,y}M_{xy}=0$, while summing all entries by columns (Bob odd) gives $\bigoplus_{x,y}M_{xy}\equiv1+1+1\equiv1$. This results in a
contradiction. Hence, at least one of the nine equally likely cells is lost, so $\wc\le 8/9$. Thus, such a strategy attains $8/9$, and it is not possible to improve on this, which we can simply check by enumerating all such possible matrices.
\end{proof}

\noindent The quantum strategy uses the \emph{Peres square} of two-qubit $\pm1$ observables
\begin{equation}\label{eq:peres}
(P_{xy})=
\begin{array}{|c|c|c|}\hline
\id\otimes Z & Z\otimes\id & Z\otimes Z\\\hline
X\otimes\id & \id\otimes X & X\otimes X\\\hline
-X\otimes Z & -Z\otimes X & Y\otimes Y\\\hline
\end{array}
\end{equation}
($x$ = row, $y$ = column).

\begin{lemma}[Peres relations]\label{lem:peres}
In \eqref{eq:peres}, the three observables in each row commute and multiply to
$+\,\id\otimes\id$; the three in each column commute and multiply to $-\,\id\otimes\id$.
\end{lemma}

\begin{proof}
Commutation is checked by counting anti-commuting tensor factors (an even number implies that the observables commute). This can be verified by a straightforward calculation.
\end{proof}

\begin{proposition}[Quantum value of $\MSG$]\label{prop:msgq}
$\wq(\MSG)=1$. This is achieved by a strategy using two EPR pairs $\Bell{+}^{\otimes2}$: Alice measures the row-$x$ observables of \eqref{eq:peres}, Bob the conjugated column-$y$ observables.
\end{proposition}

\begin{proof}
By Lemma~\ref{lem:peres}, Alice's row product is $+\id$, so her three outcomes have even parity; Bob's conjugated column product is $-\id$, giving odd parity. We can check consistency of outputs on the intersecting cell by the ricochet identity $(M\otimes\id)\Bell{+}=(\id\otimes M^{T})\Bell{+}$, also using $\overline P=P^{T}$ for a Hermitian $P$, which are their observables, gives us the identity $(P_{xy}\otimes\overline{P_{xy}})\Bell{+}^{\otimes2}=\Bell{+}^{\otimes2}$, so the measurement outcomes of their observables satisfy $a_y=b_x$ with probability one.
\end{proof}


 

\subsection{Monogamy of entanglement}\label{sec:monogamy}

\begin{lemma}[A pure subsystem factorizes off]\label{lem:purefactor}
If a bipartite state $\ket\psi_{XY}$ and its subsystem state $\rho_X=\tr_Y\proj\psi$ are both pure states, then $\ket\psi=\ket{u}_X\otimes\ket{w}_Y$.
\end{lemma}
\begin{proof}
Consider the Schmidt decomposition $\ket\psi=\sum_k\sqrt{p_k}\ket{u_k}\ket{w_k}$. Then, clearly the subsystem
$\rho_X=\sum_kp_k\proj{u_k}$ is pure iff some $p_{k_0}=1$, which implies $\ket{\psi}$ is a product pure state.
\end{proof}

\begin{lemma}[Maximal entanglement is monogamous]\label{lem:monogamy}
Let $\ket\psi_{ABC}$ be a pure state and suppose $\cH_A=\cH_{A_0}\otimes\cH_{A_1}$ such that $\rho_{A_0B}$ is a maximally entangled pure state between the two registers. Then, register $A_0$ must be uncorrelated with $C$.
\end{lemma}

\begin{proof}
Since $\rho_{A_0B}$ is a pure state, we can apply Lemma~\ref{lem:purefactor} across the cut $(A_0B\,|\,A_1C)$. So, we can write $\ket\psi=\ket\phi_{A_0B}\otimes\ket\zeta_{A_1C}$. Then, it must be true that $\rho_{A_0C}=\rho_{A_0}\otimes\rho_C$, thus the entanglement between
$A_0$ and $C$ is zero. 
\end{proof}

\noindent In other  words, if a two-qubit subsystem $A$ is maximally entangled with register $B$ then it must share no entanglement with another register $C$.

\subsection{Security definition for DI-CKA}\label{sec:secdef}

\begin{definition}[DI-CKA security \cite{ribeiro2018fully}]\label{def:cka}
An $N$-party DI-CKA protocol outputs a key $K\in\{0,1\}^\ell$ for all parties, or $\bot$. It
is $\eps$-secure if:
\begin{enumerate}
\item[(i)] \emph{Correctness} ($\eps_{\mathrm{cor}}$):
$\Pr[\text{no abort}\wedge\exists i\ne j:K_i\ne K_j]\le\eps_{\mathrm{cor}}$;
\item[(ii)] \emph{Secrecy} ($\eps_{\mathrm{sec}}$): conditioned on no abort,
$\tfrac12\big\|\rho_{KE}-\tau_K\otimes\rho_E\big\|_1\le\eps_{\mathrm{sec}}$, with $\tau_K$
maximally mixed and $E$ the adversary's quantum side information including the public
transcript,
\end{enumerate}
and $\eps=\eps_{\mathrm{cor}}+\eps_{\mathrm{sec}}$.
\end{definition}

%% file: parity_MSG1.tex
\section{\texorpdfstring{Parity Extensions of Two-Player Games.}{Parity Extensions of Two-Player Games.}}\label{sec:game}

In this section we introduce \emph{Parity-$G$ games}, an $N$-player extension of a two-player non-local game $G$. The construction applies to every $G$ that has a suitable Pauli strategy. We develop the general framework in Sections~\ref{sec:pauli-games}--\ref{sec:parityG-values}, and instantiate it with the Magic Square game in Section~\ref{sec:pmsg}. 

There may also be various other potential paths to designing a multiparty game based on two-party games for conference key agreement. We motivate our Parity-$G$ framework by showing in Appendix \ref{sec:obstruction} why a certain natural extension, for the Magic Square game in particular, cannot be used as a basis for a conference key agreement. This informs our construction of the Parity-$G$ framework. We first describe the conditions we need from a two-player game $G$ in order to extend it to an $N$-player Parity-$G$ game. We note that these conditions are very natural and satisfied by nearly all two-party games studied in the literature. 
\subsection{Two-player games with Pauli strategies}\label{sec:pauli-games}

\paragraph{Games.} Throughout this section, $G=(\pi,\cX,\cY,\{0,1\}^{L},\cB,V)$ refers to a two-player game such that
$(x,y)\sim\pi$ on the cross product of input sets $\cX\times\cY$, Alice receives $x$ and outputs a bit string $a=(a_\ell)_{\ell\in L}\in\{0,1\}^{L}$ indexed by a finite set $L$, Bob receives $y$ and outputs $b\in\cB$, and they win iff $V(x,y,a,b)=1$. Alice's output set is all of $\{0,1\}^L$, and any constraint on her answer is encoded in the winning condition $V$.

\paragraph{Pauli strings.} A \emph{signed Pauli string} on $k$ qubits is an operator $\tau=\pm\,\tau_1\otimes\cdots\otimes\tau_k$ with $\tau_i\in\{\id,X,Y,Z\}$. Its \emph{$X$-support} is the vector $\Xsupp(\tau)\in\Ftwo^k$ with $\Xsupp(\tau)_i=1$ iff $\tau_i\in\{X,Y\}$. We call $\tau$ \emph{$Z$-type} if $\Xsupp(\tau)=0$, that is, $\tau=\pm Z^{v}$ for some $v\in\Ftwo^k$, where $Z^{v}\coloneq Z^{v_1}\otimes\cdots\otimes Z^{v_k}$. We also write
$\langle u,p\rangle\coloneq\bigoplus_{i}u_ip_i$ for $u,p\in\Ftwo^k$.

\begin{lemma}[Twist rule]\label{lem:twist}
For every signed Pauli string $\tau$ on $k$ qubits and every $p\in\Ftwo^k$,
\[
Z^{p}\,\tau\,Z^{p}=(-1)^{\langle p,\Xsupp(\tau)\rangle}\,\tau .
\]
\end{lemma}

\begin{proof}
$Z$ anticommutes with $X$ and $Y$ and commutes with $\id$ and $Z$. Hence, for each $i$,
\[
Z^{p_i}\tau_iZ^{p_i}=(-1)^{p_i\Xsupp(\tau)_i}\,\tau_i ,
\]
and the claim follows by taking the tensor product over $i$.
\end{proof}

\begin{definition}[Pauli strategy and key input]\label{def:pauli-strategy}
A \emph{$k$-qubit Pauli strategy} $\mathcal S$ for a game $G$ is a quantum strategy satisfying conditions (C1), (C2) and (C3) as follows.
\begin{enumerate}
\item[\textnormal{(C1)}] \emph{Resource state.} Alice and Bob share $k$ EPR pairs $\Bell{+}^{\otimes k}$ between them.
\item[\textnormal{(C2)}] \emph{Alice measures Pauli strings.} For every input $x\in\cX$, Alice jointly measures pairwise commuting signed Pauli strings $(Q_{x,\ell})_{\ell\in L}$ on her $k$ qubits. So, she reports $(-1)^{a_\ell}$ as the outcome of $Q_{x,\ell}$. In particular, she performs the projective measurement
\begin{equation}\label{eq:pauli-proj}
\Pi^{x}_{a}\coloneq\prod_{\ell\in L}\tfrac12\big(\id+(-1)^{a_\ell}Q_{x,\ell}\big),
\qquad a\in\{0,1\}^L .
\end{equation}
\end{enumerate}
Bob's measurements are arbitrary POVMs $\{N^{y}_{b}\}_{b\in\cB}$ on his $k$ qubits. We write
$\omega(\mathcal S)$ for the winning probability of $\mathcal S$ in $G$. Also, an input $x^*\in\cX$ is a \emph{key input} for $\mathcal S$ if the following condition is met.
\begin{enumerate}
\item[\textnormal{(C3)}] \emph{Key input.} Every $Q_{x^*,\ell}$ is $Z$-type, say $Q_{x^*,\ell}=(-1)^{s_\ell}Z^{v_\ell}$, such that the vectors $(v_\ell)_{\ell\in L}$ span $\Ftwo^k$.
\end{enumerate}
\end{definition}

\begin{remark}\label{rem:conditions}
Condition (C1) says the resource state required by $\mathcal S$ should be a tensor product of EPR pairs. In our $N$-player Parity-$G$ game, this will be obtained from sharing $k$ GHZ states. Condition (C2) is the key requirement of our framework. By Lemma~\ref{lem:twist}, conjugating any $Q_{x,\ell}$ by $Z^{p}$ gives $\pm Q_{x,\ell}$. Applying $Z^{p}$ to Alice's qubits before she measures therefore leaves her measurement \eqref{eq:pauli-proj} unchanged as a set of projectors and only changes which outcome
label each projector carries (Lemma~\ref{lem:twisted-meas} makes this precise). As we will see soon, the referee of a Parity-$G$ game will need to undo such a relabelling. Note that for observables outside the Pauli group this fails: $Z\,\tfrac{X+Z}{\sqrt2}\,Z=\tfrac{-X+Z}{\sqrt2}$ is neither $\tfrac{X+Z}{\sqrt2}$ nor its negation. Condition (C2) restricts only Alice, because in our construction the operators $Z^{p}$ act on her
qubits alone; Bob's measurements are arbitrary. Condition (C3) singles out an input on which
Alice's measurement is a computational-basis measurement (Lemma~\ref{lem:parityG-key}); this
input will produce the key.
\end{remark}

\subsection{\texorpdfstring{The Parity-$G$ game}{The Parity-G game}}\label{sec:parityG}

Consider a two-player game $G$ which has a $k$-qubit Pauli strategy $\mathcal S$. We first establish some more notation by defining the \emph{swap mask} of $\mathcal S$, which intuitively the referee of the Parity-$G$ game will use to relabel Alice's answer.

\begin{definition}[Swap mask]\label{def:swap-mask}
For each input $x\in\cX$ and output index $\ell\in L$, the \emph{mask vector} of $Q_{x,\ell}$ is
$u_{x,\ell}\coloneq\Xsupp(Q_{x,\ell})\in\Ftwo^k$. We define the \emph{swap mask} of
$\mathcal S$ as the map $\sgn\colon\cX\times\Ftwo^k\to\{0,1\}^L$ given by
\begin{equation}\label{eq:mask-general}
\sgn(x,p)\coloneq\big(\sgn_\ell(x,p)\big)_{\ell\in L},\qquad
\sgn_\ell(x,p)\coloneq\langle u_{x,\ell},p\rangle .
\end{equation}
\end{definition}

By Lemma~\ref{lem:twist}, $\sgn_\ell(x,p)=1$ exactly when $Z^{p}$ anticommutes with
$Q_{x,\ell}$. Thus $\sgn(x,p)$ lists the outcome bits of Alice's measurement on input $x$ that flip when her qubits carry the twist $Z^p$. Section~\ref{sec:honeststrategy} shows that entanglement swapping produces such a twist.

\begin{definition}[Parity-$G$ game]\label{def:parityG}
For $N\ge2$, the $N$-player game $\pG{G}_N$ has players Alice, $B_1$, and parity Bobs
$B_2,\dots,B_{N-1}$.
\begin{itemize}
\item \textbf{Inputs.} The referee samples $(x,y)\sim\pi$ and sends $x$ to Alice and $y$ to
$B_1$. The parity Bobs receive no input.
\item \textbf{Outputs.} Alice outputs $a\in\{0,1\}^L$ and $B_1$ outputs $b\in\cB$. Each parity
Bob $B_j$ outputs a vector $c_j=(c_j^{(1)},\dots,c_j^{(k)})\in\Ftwo^k$, one bit per index $i\in\{1,\dots,k\}$.
\item \textbf{Parity.} The referee computes the \emph{parity vector}
$p=(p_1,\dots,p_k)\in\Ftwo^k$, where
\[
p_i\coloneq\bigoplus_{j=2}^{N-1}c_j^{(i)}\qquad(1\le i\le k)
\]
is the XOR of the parity Bobs' $i$-th bits. In vector notation,
$p=c_2\oplus\cdots\oplus c_{N-1}$ with $\oplus$ taken coordinate wise.
\item \textbf{Winning predicate.} The players win iff
\begin{equation}\label{eq:parityG-win}
V\big(x,\,y,\,a\oplus\sgn(x,p),\,b\big)=1 .
\end{equation}
\end{itemize}
\end{definition}

The game depends on $\mathcal S$ only through the mask vectors, and we suppress this dependence
in the notation. For $N=2$ there are no parity Bobs, $p=0$, and $\pG{G}_2=G$. Since Alice's output set is all of $\{0,1\}^L$, the relabeled answer $a\oplus\sgn(x,p)$ is again an element of it, so the predicate \eqref{eq:parityG-win} is well defined. This is why constraints on Alice's answer are encoded in $V$. 

\begin{lemma}[Properties of the mask]\label{lem:mask-props}
Given a game $G$ with a $k$-qubit Pauli strategy $\mathcal{S}$, the swap mask $\sigma$ has the following properties.
\begin{enumerate}
\item[\textnormal{(i)}] $\sgn(x,\cdot)$ is linear:
$\sgn(x,p\oplus p')=\sgn(x,p)\oplus\sgn(x,p')$, and $\sgn(x,0)=0$.
\item[\textnormal{(ii)}] If $x^*$ is a key input, then $\sgn(x^*,p)=0$ for every $p$.
\item[\textnormal{(iii)}] If $T\subseteq L$ and $\prod_{\ell\in T}Q_{x,\ell}=\pm\id$, then
$\bigoplus_{\ell\in T}\sgn_\ell(x,p)=0$ for every $p$.
\end{enumerate}
\end{lemma}

\begin{proof}
(i) follows directly from \eqref{eq:mask-general}, and (ii) holds because $Z$-type strings have zero $X$-support. For (iii), inserting $Z^pZ^p=\id$ between the factors and applying Lemma~\ref{lem:twist} gives
$Z^{p}\big(\prod_{\ell\in T}Q_{x,\ell}\big)Z^{p}
=(-1)^{\bigoplus_{\ell\in T}\sgn_\ell(x,p)}\prod_{\ell\in T}Q_{x,\ell}$.
The left side equals $\prod_{\ell\in T}Q_{x,\ell}$ because this product is $\pm\id$.
\end{proof}

Part (iii) says that the mask respects every multiplicative relation among Alice's observables.
In particular, the relabelling never breaks a parity constraint that $\mathcal S$ satisfies: if
$\prod_{\ell\in T}Q_{x,\ell}=(-1)^{s}\id$, then Alice's answers under $\mathcal S$ satisfy
$\bigoplus_{\ell\in T}a_\ell=s$, and so does $a\oplus\sgn(x,p)$ for every $p$.

\subsection{The honest strategy}\label{sec:honeststrategy}

The honest strategy for a Parity-$G$ game can be constructed by applying the corresponding two player strategy $\mathcal S$ for $G$ on $k$ EPR pairs shared between Alice and $B_1$, with some additional measurements by the parity Bobs.

\paragraph{The strategy.} The $N$ players share $\ket{\GHZ_N}^{\otimes k}$. Every player holds one qubit of each copy, and $P^{(i)}$ denotes player $P$'s qubit of copy $i$. They carry out the following measurements.
\begin{enumerate}
\item[\textnormal{(M1)}] Each parity Bob $B_j$ measures $X$ on each of their $k$ qubits, obtaining the bit $c_j^{(i)}$ from copy $i$, and outputs $c_j=(c_j^{(1)},\dots,c_j^{(k)})$.
\item[\textnormal{(M2)}] On input $x$, Alice performs $\{\Pi^{x}_{a}\}_a$ of
\eqref{eq:pauli-proj} on $(A^{(1)},\dots,A^{(k)})$ and outputs $a$.
\item[\textnormal{(M3)}] On input $y$, $B_1$ performs $\{N^{y}_{b}\}_b$ on
$(B_1^{(1)},\dots,B_1^{(k)})$ and outputs $b$.
\end{enumerate}

\begin{lemma}[Single-copy swap]\label{lem:swap}
Consider a single copy of $\ket{\GHZ_N}$ and group its qubits as
\[
\ket{\GHZ_N}=\tfrac1{\sqrt2}\big(\ket0_A\ket0_{B_1}\ket{0\cdots0}_{\mathrm{par}}
+\ket1_A\ket1_{B_1}\ket{1\cdots1}_{\mathrm{par}}\big),
\]
where the register $\mathrm{par}$ holds the qubits of the parity Bobs $B_2,\dots,B_{N-1}$. Each parity Bob $B_j$ measures their qubit in the $X$ basis, obtaining the bit $e_j\in\{0,1\}$. Let
\[
e\coloneq\bigoplus_{j=2}^{N-1}e_j .
\]
Then, each outcome string $(e_2,\dots,e_{N-1})\in\{0,1\}^{N-2}$ occurs with probability $2^{-(N-2)}$, and conditioned on it, the state of $(A,B_1)$ collapses to
\begin{equation}\label{eq:singleswap}
\tfrac1{\sqrt2}\big(\ket{00}+(-1)^{e}\ket{11}\big)=(Z^{e}\otimes\id)\Bell{+}.
\end{equation}
For $N=2$ there are no parity Bobs, $e=0$, and the state is $\Bell{+}=\ket{\GHZ_2}$.
\end{lemma}

\begin{proof}
After all the parity Bobs $B_2,...,B_{N-1}$ measure the GHZ state in the $X$ basis, with $B_j$ measuring the bit $e_j$, we can check that the state collapses to the unnormalized two-qubit state
\[
\tfrac1{\sqrt2}\Big(\underbrace{\textstyle\prod_{j}\tfrac1{\sqrt2}}_{(\sqrt2)^{-(N-2)}}\ket{00}
+\underbrace{\textstyle\prod_{j}\tfrac{(-1)^{e_j}}{\sqrt2}}_{(-1)^{e}(\sqrt2)^{-(N-2)}}\ket{11}\Big)
=(\sqrt2)^{-(N-2)}\cdot\tfrac1{\sqrt2}\big(\ket{00}+(-1)^{e}\ket{11}\big),
\]
where $\prod_j(-1)^{e_j}=(-1)^{e}$ because only the parity of the number of minus signs matters. Normalizing this, we obtain the state in Eq. \eqref{eq:singleswap}.
\end{proof}

\begin{lemma}[Dressed EPR resource]\label{lem:dressed}
Under the measurement \textnormal{(M1)}, every outcome $(c_2,\dots,c_{N-1})$ occurs with probability
$2^{-k(N-2)}$, so $p$ is uniform on $\Ftwo^k$. Conditioned on the outcome, the state of
Alice's and $B_1$'s qubits is
\begin{equation}\label{eq:dressed}
\ket{\Psi_p}=\big(Z^{p}_{A}\otimes\id_{B_1}\big)\,\Bell{+}^{\otimes k},\qquad
Z^{p}_{A}=Z^{p_1}_{A^{(1)}}\otimes\cdots\otimes Z^{p_k}_{A^{(k)}},
\end{equation}
where the $i$-th EPR pair is on $(A^{(i)},B_1^{(i)})$.
\end{lemma}

\begin{proof}
We apply Lemma~\ref{lem:swap} to each copy. On copy $i$, the parity Bobs' outcome bits are $e_j=c_j^{(i)}$, so $e=p_i$; copy $i$ therefore yields $(Z^{p_i}\otimes\id)\Bell{+}$ on $(A^{(i)},B_1^{(i)})$, and its outcome string has probability $2^{-(N-2)}$. The lemma follows from taking the product over $i$.
\end{proof}

\begin{lemma}[The twist relabels Alice's outcome]\label{lem:twisted-meas}
For all $x\in\cX$, $a\in\{0,1\}^L$ and $p\in\Ftwo^k$,
$Z^{p}\,\Pi^{x}_{a}\,Z^{p}=\Pi^{x}_{a\oplus\sgn(x,p)}$.
\end{lemma}

\begin{proof}
By Lemma~\ref{lem:twist}, conjugating each factor of \eqref{eq:pauli-proj} gives
$Z^{p}\tfrac12\big(\id+(-1)^{a_\ell}Q_{x,\ell}\big)Z^{p}
=\tfrac12\big(\id+(-1)^{a_\ell\oplus\sgn_\ell(x,p)}Q_{x,\ell}\big)$.
\end{proof}

\begin{proposition}[Honest winning probability]\label{prop:parityG-honest}
For every $p\in\Ftwo^k$, conditioned on the parity being $p$, the honest strategy
\textnormal{(M1)--(M3)} wins $\pG{G}_N$ with probability $\omega(\mathcal S)$. In particular it
wins with probability $\omega(\mathcal S)$.
\end{proposition}

\begin{proof}
Fix $p$ and $(x,y)$, and let $P_{\mathcal S}(a,b\mid x,y)$ be the output distribution of
$\mathcal S$ in $G$. By Lemmas~\ref{lem:dressed} and~\ref{lem:twisted-meas},
\[
\Pr[a,b\mid x,y,p]
=\bra{\Phi^+}^{\otimes k}\big(Z^{p}\Pi^{x}_{a}Z^{p}\otimes N^{y}_{b}\big)\Bell{+}^{\otimes k}
=P_{\mathcal S}\big(a\oplus\sgn(x,p),\,b\mid x,y\big).
\]
So, $(a\oplus\sgn(x,p),b)$ has the same distribution as the outputs of $\mathcal S$, and the referee evaluates $V$ on this pair. Averaging over $(x,y)\sim\pi$ gives the winning probability $\omega(\mathcal S)$.
\end{proof}

\subsection{Quantum and classical values}\label{sec:parityG-values}

A quantum strategy for $\pG{G}_N$ consists of a state $\rho$ on
$Q_A\otimes Q_{B_1}\otimes Q_{\mathrm{par}}$, where
$Q_{\mathrm{par}}=Q_{B_2}\otimes\cdots\otimes Q_{B_{N-1}}$, together with POVMs
$\{M^{x}_{a}\}$ for Alice, $\{N^{y}_{b}\}$ for $B_1$, and $\{D^{(j)}_{c}\}_{c\in\Ftwo^k}$ for
each parity Bob, which carries no input label. For $p\in\Ftwo^k$ let
\[
D_p\coloneq\sum_{\substack{c_2,\dots,c_{N-1}\in\Ftwo^k\\ c_2\oplus\cdots\oplus c_{N-1}=p}}\ \bigotimes_{j=2}^{N-1}D^{(j)}_{c_j},\qquad
q_p\coloneq\tr\big[\rho\,(\id\otimes\id\otimes D_p)\big],
\]
and, for $q_p>0$, let
$\rho^{(p)}\coloneq q_p^{-1}\tr_{\mathrm{par}}\big[(\id\otimes\id\otimes D_p^{1/2})\,\rho\,
(\id\otimes\id\otimes D_p^{1/2})\big]$ be the state of $Q_AQ_{B_1}$ conditioned on parity $p$.

\begin{theorem}[Values of Parity-$G$ games]\label{thm:parityG-values}
For every $N\ge2$,
\[
\wq(\pG{G}_N)=\wq(G)\qquad\text{and}\qquad\wc(\pG{G}_N)=\wc(G).
\]
The honest strategy attains $\omega(\mathcal S)$, so it is optimal whenever $\mathcal S$ is
optimal for $G$.
\end{theorem}

\begin{proof}
\emph{Lower bound.} If every parity Bob outputs $c_j=0$, then $p=0$, $\sgn(\cdot,0)=0$, and
Alice and $B_1$ play $G$. Every quantum or classical strategy for $G$ therefore wins
$\pG{G}_N$ with the same probability.

\emph{Quantum upper bound.} The parity Bobs have no input, so $D_p$ does not depend on
$(x,y)$ and
\[
\Pr[\mathrm{win}]=\sum_{p}q_p\,\omega_p,\qquad
\omega_p\coloneq\sum_{x,y}\pi(x,y)\sum_{a,b}\tr\big[\rho^{(p)}(M^{x}_{a}\otimes N^{y}_{b})\big]\,
V\big(x,y,a\oplus\sgn(x,p),b\big).
\]
Substituting $a'=a\oplus\sgn(x,p)$, which is a bijection of $\{0,1\}^L$, shows that $\omega_p$ is the winning probability in $G$ of
the strategy with state $\rho^{(p)}$, Alice's POVMs $M'^{x}_{a'}\coloneq M^{x}_{a'\oplus\sgn(x,p)}$
and Bob's POVMs $N^{y}_{b}$. Hence $\omega_p\le\wq(G)$ for every $p$, and
$\Pr[\mathrm{win}]\le\wq(G)$.

\emph{Classical upper bound.} By convexity it suffices to bound deterministic strategies. The
parity Bobs then output constants, so $p$ is fixed, and the predicate sees Alice's answer
function $x\mapsto f(x)$ as $x\mapsto f(x)\oplus\sgn(x,p)$, which is again a deterministic
answer function. The winning probability is that of a deterministic strategy for $G$, at most
$\wc(G)$.

The last claim is Proposition~\ref{prop:parityG-honest}.
\end{proof}

\begin{corollary}\label{cor:parityG-pt}
If $G$ is a pseudo-telepathy game and $\omega(\mathcal S)=1$, then $\pG{G}_N$ is a
pseudo-telepathy game for every $N\ge2$, and the honest strategy wins it with certainty.
\end{corollary}

\subsection{The Parity Magic Square game}\label{sec:pmsg}

We now study the Parity-$G$ game for the special case when $G$ is the Magic Square game. For the rest of the paper, we will mostly focus on building a conference key protocol from this game and analyzing its security, to simplify exposition. When needed, we will show how our analysis carries over to general Parity-$G$ games.

So, let $G=\MSG$ with $L=\{0,1,2\}$, $\cX=\cY=\{0,1,2\}$, $\pi$ uniform, and Alice's parity constraint folded into the predicate:
\[
V_{\MSG}(x,y,a,b)=1\iff
\textstyle\bigoplus_{\ell=0}^{2}a_\ell=0,\quad\bigoplus_{\ell=0}^{2}b_\ell=1,\quad a_y=b_x .
\]
Answers violating a parity constraint lose in this formulation and in the one of the
preliminaries, so the two define the same game.

\paragraph{The Pauli strategy.} Let $\mathcal S_{\mathrm{P}}$ be the optimal quantum strategy (or Peres strategy) described in
Proposition~\ref{prop:msgq}: Alice holds $(A^{(1)},A^{(2)})$ of $\Bell{+}^{\otimes2}$ and on row
$x$ measures $Q_{x,y}\coloneq P_{xy}$ from \eqref{eq:peres}, and $B_1$ on column $y$ measures
$\overline{P_{0y}},\overline{P_{1y}},\overline{P_{2y}}$. We note that $\mathcal S_{\mathrm{P}}$ is a valid Pauli strategy since it satisfies the conditions of Definition~\ref{def:pauli-strategy} with $k=2$:
\begin{itemize}
\item (C1) holds with two EPR pairs.
\item (C2) holds because every entry of \eqref{eq:peres} is a signed Pauli string and the
entries of each row commute (Lemma~\ref{lem:peres}).
\item (C3) holds with $x^*=0$: the row-$0$ observables are $\id\otimes Z=Z^{(0,1)}$,
$Z\otimes\id=Z^{(1,0)}$ and $Z\otimes Z=Z^{(1,1)}$, and $(0,1),(1,0)$ span $\Ftwo^2$.
\end{itemize}
Moreover, $\omega(\mathcal S_{\mathrm P})=1$ by Proposition~\ref{prop:msgq}.

\paragraph{The swap mask.} Writing $P_{xy}=\pm\,\tau^{(1)}_{xy}\otimes\tau^{(2)}_{xy}$, the mask vector of cell $(x,y)$ is
$u_{x,y}=(\alpha_{xy},\beta_{xy})$, where
\begin{equation}\label{eq:ab}
\alpha_{xy}\coloneq[\tau^{(1)}_{xy}\in\{X,Y\}],\qquad
\beta_{xy}\coloneq[\tau^{(2)}_{xy}\in\{X,Y\}] .
\end{equation}
Table~\ref{tab:mask} lists these values. By \eqref{eq:mask-general}, the mask on cell $(x,y)$
is $\alpha_{xy}p_1\oplus\beta_{xy}p_2$; following the rest of the paper, we write it as
$\sgn(x,y)$ and leave the dependence on $p=(p_1,p_2)$ implicit.

\begin{table}[h]
\centering
\renewcommand{\arraystretch}{1.15}
\begin{tabular}{c|ccc}
\rowcolor{headgrey}
$(\alpha_{xy},\beta_{xy})$ & $y=0$ & $y=1$ & $y=2$\\
\hline
$x=0$ & $(0,0)$ & $(0,0)$ & $(0,0)$\\
$x=1$ & $(1,0)$ & $(0,1)$ & $(1,1)$\\
$x=2$ & $(1,0)$ & $(0,1)$ & $(1,1)$\\
\end{tabular}
\caption{The mask vectors \eqref{eq:ab} of the Peres strategy \eqref{eq:peres}. The mask is
$\sgn(x,y)=\alpha_{xy}p_1\oplus\beta_{xy}p_2$.}
\label{tab:mask}
\end{table}

\begin{definition}[Parity Magic Square game $\pMSG_N$]\label{def:pmsg}
$\pMSG_N$ is defined as the Parity-$G$ game, from Definition~\ref{def:parityG}, with $G=\MSG$ and the Peres
strategy $\mathcal S_{\mathrm P}$. Explicitly, the players are Alice, Bob $B_1$, and parity Bobs $B_2,\dots,B_{N-1}$.
\begin{itemize}
\item \textbf{Inputs.} Alice receives a row $x\in\{0,1,2\}$ and $B_1$ a column
$y\in\{0,1,2\}$, both input are uniform random and independent. The parity Bobs receive no input.
\item \textbf{Outputs.} Alice outputs $a=(a_0,a_1,a_2)$, $B_1$ outputs $b=(b_0,b_1,b_2)$, and
each parity Bob $B_j$ outputs two bits $(c_j^{(1)},c_j^{(2)})$.
\item \textbf{Parities and mask.} The referee sets
$p_1=\bigoplus_{j=2}^{N-1}c_j^{(1)}$, $p_2=\bigoplus_{j=2}^{N-1}c_j^{(2)}$, and the mask
\begin{equation}\label{eq:mask}
\sgn(0,y)=0\ (\forall y),\qquad
\sgn(x,0)=p_1,\ \ \sgn(x,1)=p_2,\ \ \sgn(x,2)=p_1\oplus p_2\quad(x\in\{1,2\}).
\end{equation}
\item \textbf{Winning predicate.}
\begin{equation}\label{eq:pmsg-win}
\bigoplus_{\ell=0}^{2}a_\ell=0\ \ (\text{Alice even}),\quad
\bigoplus_{\ell=0}^{2}b_\ell=1\ \ (B_1\text{ odd}),\quad
a_y\oplus b_x=\sgn(x,y)\ \ (\text{masked cell check}).
\end{equation}
\end{itemize}
For $N=2$ there are no parity Bobs, $\sgn\equiv0$, and $\pMSG_2=\MSG$.
\end{definition}

\begin{remark}[Equivalence with Definition~\ref{def:parityG}]\label{rem:relabel}
The general predicate \eqref{eq:parityG-win} checks Alice's parity on $a\oplus\sgn(x,\cdot)$,
while \eqref{eq:pmsg-win} checks it on $a$. The two agree: each row of the Peres square
multiplies to $+\id$ (Lemma~\ref{lem:peres}), so by Lemma~\ref{lem:mask-props} (iii)
$\bigoplus_{y}\sgn(x,y)=0$ for every $x$, which is also implied by \eqref{eq:mask}. Hence, the referee may equivalently replace Alice's cell value $a_y$ by $a_y\oplus\sgn(x,y)$ and apply the
ordinary magic-square check. Note that the relabeled triple would still have even parity.
\end{remark}

\paragraph{The honest strategy.} Here, we simply instantiate the measurement strategy \textnormal{(M1)--(M3)} from Section \ref{sec:honeststrategy} with $k=2$:
\begin{enumerate}
\item[\textnormal{(S1)}] Each parity Bob $B_j$ measures $X$ on both of their qubits and gets $c_j^{(1)},c_j^{(2)}$.
\item[\textnormal{(S2)}] On row input $x$, Alice measures the commuting Peres observables $P_{x0},P_{x1},P_{x2}$ of \eqref{eq:peres} jointly on her qubits $(A^{(1)},A^{(2)})$, getting the measurement outcomes $a_0,a_1,a_2$.
\item[\textnormal{(S3)}] On column input $y$, $B_1$ measures the commuting conjugated observables $\overline{P_{0y}},\overline{P_{1y}},\overline{P_{2y}}$ jointly on his qubits $(B_1^{(1)},B_1^{(2)})$, getting the measurement outcomes $b_0,b_1,b_2$.
\end{enumerate}
By Lemma~\ref{lem:dressed}, after (S1), Alice and $B_1$ obtain the state
$(Z^{p_1}_{A^{(1)}}\otimes Z^{p_2}_{A^{(2)}})\Bell{+}^{\otimes2}$, and by
Lemma~\ref{lem:twisted-meas} the twist flips exactly the cells with $\sgn(x,y)=1$. We also note the following properties of $\pMSG_N$.

\begin{proposition}[Quantum value]\label{prop:wq}
$\wq(\pMSG_N)=1$: the strategy \textnormal{(S1)--(S3)} wins on every input and for every
outcome of the parity Bobs.
\end{proposition}

\begin{proof}
This follows from Proposition~\ref{prop:parityG-honest} with $\omega(\mathcal S_{\mathrm P})=1$.
\end{proof}

\begin{proposition}[Classical value]\label{prop:wc}
$\wc(\pMSG_N)=8/9<1$. Hence $\pMSG_N$ is a pseudo-telepathy game.
\end{proposition}

\begin{proof}
This follows from Theorem~\ref{thm:parityG-values} and $\wc(\MSG)=8/9$.
\end{proof}

\begin{remark}
The classical value can also be checked directly. For a deterministic strategy the mask is
fixed, and the row-versus-column parity count of the magic square acquires the extra term
$\bigoplus_{x,y}\sgn(x,y)=p_1\bigoplus_{x,y}\alpha_{xy}\oplus p_2\bigoplus_{x,y}\beta_{xy}$.
Table~\ref{tab:mask} has four ones in each coordinate, so this term vanishes for every $(p_1,p_2)$, and we can follows the same proof as the 2-player Magic Square game.
\end{remark}

\paragraph{Key input.} The key input for Alice is $x^*=0$, which is never masked.

\subsection{Further instances}\label{sec:parityG-examples}
Here, we show some examples of other games which also fit into the framework of Parity-$G$ games.
\begin{example}[Parity-CHSH]\label{ex:pchsh}
Let $G=\mathsf{CHSH}$: $x,y\in\{0,1\}$ uniform, $|L|=1$, $b\in\{0,1\}$, and
$V(x,y,a,b)=[\,a\oplus b=xy\,]$. The optimal strategy on one EPR pair has Alice measure
$Q_0=Z$ and $Q_1=X$ and Bob measure $\tfrac{1}{\sqrt2}\big(Z+(-1)^{y}X\big)$. With this strategy, they win with probability $\cos^2(\pi/8)$. Alice's observables are Pauli strings and $x^*=0$ is a key input. The mask vectors are $u_0=0$ and $u_1=1$, so $\sgn(x,p)=xp$, and the resulting Parity-$G$ game, which we denote $\pCHSH_N$, is won iff
\[
a\oplus b\oplus xp=xy,\qquad p=\textstyle\bigoplus_{j=2}^{N-1}c_j .
\]
By Theorem~\ref{thm:parityG-values}, $\wq(\pCHSH_N)=\cos^2(\pi/8)$ and $\wc(\pCHSH_N)=3/4$. Bob's observables are not Pauli strings, which is allowed since the twist sits on Alice's side. The Parity-CHSH game was defined by \cite{ribeiro2018fully}, and also used as part of their conference key agreement protocol.
\end{example}

\begin{example}[Binary constraint system games]\label{ex:bcsg}
We note that a subset of \emph{Binary constraint system games} admitting Pauli strategies \cite{cleve2014characterization} can be used to construct their corresponding Parity-$G$ games as well, following a similar argument as above.
\end{example}

%% file: CKA.tex
\section{The DI-CKA protocol}\label{sec:protocol}

We now show our conference key agreement protocol, basing it on the Parity Magic Square game, followed by a proof of its security. By Remark~\ref{cor:nocons}, it seems challenging and perhaps impossible in certain cases to obtain a conference key from a single game, so our protocol uses two measurement settings depending on whether the players are in a test round or a key round. 

\subsection{Setup}

There are $N$ parties $P_1=\text{Alice},P_2,\dots,P_N=B_1,\dots,B_{N-1}$, each holding one untrusted device with settings $\{X,Y,Z\}$ and classical outputs. We make the same standard assumptions as used in DI-QKD literature (isolated parties, trusted local randomness, authenticated public broadcast given to Alice and Bob), see for instance \cite{cervero2025device}. In particular, we assume that the devices can not access classical communication between Alice and Bob during the protocol. This is important when Alice announces a round as being a test round or a key round, for instance, as we will see in the protocol. At the end, this information is made public, since their classical channels are not secret to an adversary Eve. Note that this is consistent with the non-signaling assumption between the devices and Eve, during the protocol. As we are in the device-independent setting, we assume nothing about the devices' internal workings and allow them to be maximally malicious within these constraints.

\subsection{The DI-CKA protocol}
Here we describe the conference key agreement protocol based on $\pMSG_N$.

\begin{protocolbox}[Protocol 1: $\pMSG_N$-DI-CKA] \label{pr:one}
\textbf{Parameters.} Rounds $n$; test probability $\gamma\in(0,1)$; tolerated score
$\omega_{\mathrm{tol}}=\omega(v)-\delta$; error-correction budget; correctness-hash length
$t$; key length $\ell$.

\smallskip
\textbf{For each round $i=1,\dots,n$:}
\begin{enumerate}
\item[(1)] \emph{Round type.} Alice asks her device to generate $\ket{\GHZ_N}^{\otimes2}$ and distributes it with all the players. She also broadcasts a coin $T_i$ with $\Pr[T_i=1]=\gamma$.
\item[(2a)] \emph{Key round} ($T_i=0$). Every player asks their devices to measure $Z$ and records their output, setting their raw key bits $\{K_{j,i}^b\}_{b \in \{0,1\}}$.
\item[(2b)] \emph{Test round} ($T_i=1$). Alice chooses a uniform random row $x_i$, $B_1$ chooses a uniform random column $y_i$, and they both ask their devices to perform measurements according to the cell $(x_i,y_i)$ in \eqref{eq:peres}; all the parity Bobs $B_2,...,B_{N-1}$ ask their devices to measure $X$ on both their qubits.
\end{enumerate}
(Note that all devices may behave dishonestly, both when preparing states or performing measurements) \\
\textbf{After all rounds:}
\begin{enumerate}
\item[(3)] \emph{Parameter estimation.}  All test-round inputs and outputs are broadcast. All parties calculate $p_{1,i},p_{2,i}$, the masks $\sgn_i$, and the
empirical winning frequency $\hat\omega$. If $\hat\omega<\omega_{\mathrm{tol}}$, they \textbf{abort}.
\item[(4)] \emph{Sifting.} Each party's raw key is its key-round string
$\mathbf K_j=\{K_{j,i}^0,K_{j,i}^1\}_{i:T_i=0}$.
\item[(5)] \emph{Error correction (one-to-many).} Alice broadcasts a single syndrome of
$\mathbf K_1$ of length $\leakEC\approx n_{\mathrm{key}}\max_j H(K_1|K_{B_j})+O(\log1/\eps_{\mathrm{EC}})$;
each Bob decodes towards $\mathbf K_1$. Alice broadcasts a $2$-universal hash of $\mathbf K_1$
of length $t$; abort on any mismatch.
\item[(6)] \emph{Privacy amplification.} Alice broadcasts a hash seed; all parties hash their corrected keys to length $\ell$, obtaining $K_1,\dots,K_N$.
\end{enumerate}
\end{protocolbox}

In short, every party measures $Z$ on key rounds, so all $N$ parties obtain the same key. The test rounds are just used to ensure that the devices behave honestly throughout the protocol, since they are not able to distinguish between key rounds and some test rounds (in particular, when Alice's test round input is 0).

Now, we can generalize this conference key agreement protocol to DI-CKA protocols based on Parity-$G$ games as well. Let $\pG{G}_N$ be a Parity-$G$ game corresponding to the two player game $G$. Similarly as before, we interleave test rounds of playing $\pG{G}_N$ with key rounds, in which Alice's device receives the key input $x^*$ and every other party measures (say) $k$ qubits in the computational basis. Condition (C3) from Definition \ref{def:pauli-strategy} makes Alice's key round output equivalent to a $Z$-basis outcome (or a bitstring), so all parties obtain a common
string. So, the security of this string can be certified by passing the test rounds, as we explain later in the analysis in the next section.
 
\begin{lemma}[Key input and consensus]\label{lem:parityG-key}
Let $x^*$ be a key input for $\mathcal S$, with $Q_{x^*,\ell}=(-1)^{s_\ell}Z^{v_\ell}$ as used in Definition \ref{def:pauli-strategy}. Then,
\begin{enumerate}
\item[\textnormal{(i)}] Every nonzero $\Pi^{x^*}_{a}$ equals $\proj{z}$ for the unique
$z\in\Ftwo^k$ with $a=\big(\langle v_\ell,z\rangle\oplus s_\ell\big)_{\ell\in L}$. Hence Alice's
output on $x^*$ and her computational-basis outcome $z$ determine each other.
\item[\textnormal{(ii)}] If all $N$ players measure each of their qubits of $\ket{\GHZ_N}^{\otimes k}$ in the computational basis, they all obtain the same uniformly random $z\in\Ftwo^k$.
\end{enumerate}
\end{lemma}
 
\begin{proof}
(i) $Q_{x^*,\ell}\ket z=(-1)^{\langle v_\ell,z\rangle\oplus s_\ell}\ket z$, so $\Pi^{x^*}_a$
projects onto the span of the $\ket z$ with $\langle v_\ell,z\rangle\oplus s_\ell=a_\ell$ for all
$\ell$. Since the $v_\ell$ span $\Ftwo^k$, at most one $z$ solves these equations.
(ii) Follows directly by the entanglement structure of each GHZ state.
\end{proof}

%% file: security.tex
\section{Security proofs for $\pMSG_N$-DI-CKA}\label{sec:security-proofs}
We start by recalling the security definition for a general $N$-party DI-CKA protocol.
\begin{definition}[Restatement of Definition \ref{def:cka}]
An $N$-party DI-CKA protocol outputs a key $K\in\{0,1\}^\ell$ for all parties, or $\bot$. It
is $\eps$-secure if:
\begin{enumerate}
\item[(i)] \emph{Correctness} ($\eps_{\mathrm{cor}}$):
$\Pr[\text{no abort}\wedge\exists i\ne j:K_i\ne K_j]\le\eps_{\mathrm{cor}}$;
\item[(ii)] \emph{Secrecy} ($\eps_{\mathrm{sec}}$): conditioned on no abort,
$\tfrac12\big\|\rho_{KE}-\tau_K\otimes\rho_E\big\|_1\le\eps_{\mathrm{sec}}$, with $\tau_K$
maximally mixed and $E$ the adversary's quantum side information including the public
transcript,
\end{enumerate}
and $\eps=\eps_{\mathrm{cor}}+\eps_{\mathrm{sec}}$.
\end{definition}

We first prove a theorem showing the correctness of the $\pMSG_N$-DI-CKA protocol.

\begin{theorem}[Correctness of Protocol 1]\label{thm:correctness}
For any parameter estimation tolerance bound $\delta \in [0,1/9)$, Protocol 1 is $\eps_{\mathrm{cor}}$ for
\begin{equation}
    \eps_{\mathrm{cor}} \le (N-1)(\eps_{EC}+\eps'_{EC}) + \left(1-\gamma\left(1-e^{-2\delta^2}\right)\right)^n
\end{equation}
where $\max\limits_{j}\{\Pr[\mathrm{EC}(j) \text{ } \mathrm {aborts}]\} \le \eps_{\mathrm{EC}}$ and $\max\limits_{i,j}\{\Pr[K_i\ne K_j \mid \mathrm{no} \text{ } \mathrm{abort}]\} \le \eps'_{\mathrm{EC}}$.
\end{theorem}
\begin{proof}
    In an honest implementation of the protocol, we assume the measurement devices acted honestly, which means they prepared the correct states and reported the correct measurement outcomes to all the players. If our protocol is honestly implemented, we show that the probability of getting a bit of key per round shared by all the players is high, as follows.
    
    Conditioned on not aborting, we know that the probability of any two of them not sharing the same key bits is bounded by $\eps_{\mathrm{EC}}$. By applying a union bound on the mismatch events corresponding to Alice's key bits and each of the $N-1$ Bobs' key bits, we obtain that 
    \begin{equation}\label{eq:nokey_noabort}
        \Pr[\exists i,j: K_i\ne K_j \mid \mathrm{no} \text{ } \mathrm{abort}] \le (N-1)\eps'_{\mathrm{EC}}.
    \end{equation}
    Now, we just need to bound the probability of the protocol aborting. There are two scenarios in which the protocol aborts.
    \begin{enumerate}
        \item An error correction step failed, which happens with probability at most $\eps_{\mathrm{EC}}$ per Bob's EC$(j)$. By a union bound, the probability of any one of them failing is at most
        \begin{equation}\label{eq:ECabort}
            \Pr[\exists j: \mathrm{EC}(j) \text{ } \mathrm {aborts}] \le (N-1)\eps_{\mathrm{EC}}.
        \end{equation}
        \item The parameter estimation step failed, conditioned on the protocol not aborting in the error correction step.
    \end{enumerate}
    The latter probability is bounded by the following expression. Let $G^{(j)}$ denote Alice's guess for Bob $j$'s testing-round bits, and let $G^{(1,\ldots,N-1)}$ be correct whenever $G^{(j)}$ is correct for every $j\in\{1,\ldots,N-1\}$. Then,
    \begin{align} 
    \Pr(\mathrm{PE} \text{ } \mathrm{aborts}) &= \Pr\!\left[G^{(1,\ldots,N-1)} \text{ is correct}\right] \Pr\!\left[ \sum_{i=1}^{n} C_i < \omega_{\mathrm{tol}}\sum_{i=1}^{n}T_i \,\middle|\, G^{(1,\ldots,N-1)} \text{ is correct} \right] \notag \\ 
    &\quad + \Pr\!\left[\exists j : G^{(j)} \text{ is wrong}\right] \Pr\!\left[ \sum_{i=1}^{n} C_i < \omega_{\mathrm{tol}}\sum_{i=1}^{n}T_i \,\middle|\, \exists j : G^{(j)} \text{ is wrong} \right] . 
    \end{align}

    Here $T_i$ denotes the indicator random variable for the event that round $i$ is chosen for testing, so that the $T_i$ are i.i.d. Bernoulli random variables with $\Pr[T_i=1]=\gamma$. Moreover, $C_i$ denotes the indicator random variable for the event that the parameter estimation test passed in test round $i$. So, for every $j\in\{0,\ldots,n\}$,
    \begin{equation}
        \Pr\!\left[\sum_{i=1}^{n}T_i=j\right] = \binom{n}{j}(1-\gamma)^{n-j}\gamma^j . 
    \end{equation}
    Now, conditioned on $\sum_{i=1}^{n}T_i=j$ and assuming $\mathbb{E}[C_i]=\omega(\nu)$ such that $\delta = \omega(v)-\omega_{\mathrm{tot}}$, Hoeffding's inequality gives 
    \[ 
    \Pr\!\left[ \sum_{i=1}^{n}C_i < (\omega_{\mathrm{tol}}) j \,\middle|\, \sum_{i=1}^{n}T_i=j \right] \leq e^{-2\delta^2 j}. 
    \]
    We can also trivially bound $ \Pr\!\left[G^{(1,\ldots,N-1)} \text{ is correct}\right]\leq 1$. Substituting back into the previous expression, we get

    \begin{align} 
    P_{\mathrm{PE}}(\mathrm{abort}) &= \sum_{j=0}^{n} \Pr\!\left[\sum_{i=1}^{n}T_i=j\right] \Pr\!\left[ \sum_{i=1}^{n}C_i < (\omega_{\mathrm{tol}})j \,\middle|\, \sum_{i=1}^{n}T_i=j \right] \notag \\ 
    &\le \sum_{j=0}^{n} \binom{n}{j}(1-\gamma)^{n-j}\gamma^j e^{-2\delta^2 j} \notag \\ 
    &= \sum_{j=0}^{n} \binom{n}{j}(1-\gamma)^{n-j} \left(\gamma e^{-2\delta^2}\right)^j \notag \\ 
    &= \left(1-\gamma+\gamma e^{-2\delta^2}\right)^n \notag \\ 
    &= \left(1-\gamma\left(1-e^{-2\delta^2}\right)\right)^n \,, \label{eq:parabort}
    \end{align}  
    where the final equality follows from the binomial theorem as 
    \[ 
    \sum_{j=0}^{n} \binom{n}{j}(1-\gamma)^{n-j} \left(\gamma e^{-2\delta^2}\right)^j = \left(1-\gamma+\gamma e^{-2\delta^2}\right)^n = \left(1-\gamma\left(1-e^{-2\delta^2}\right)\right)^n . 
    \]
    
\noindent Thus, the total error probability is bounded by the sum of the failure probabilities of any of the three events in \eqref{eq:nokey_noabort}, \eqref{eq:ECabort} and \eqref{eq:parabort}, which proves the lemma.
\end{proof}

It remains to prove the secrecy of Protocol 1. For this, the main quantity to control is the smooth conditional min-entropy of Alice's raw key string conditioned on the adversary's quantum side information and the public transcript. More precisely, we require a lower bound on 
\begin{equation}
    H_{\min}^{\eps}\!\left(A_1^n \mid X=0, E\right)_{\rho \mid \Omega},
\end{equation}
where $\rho_{A_1^n E}$ denotes the classical-quantum state describing Alice's raw key string and the adversary's side information immediately before privacy amplification, conditioned on $\Omega$, the event that the protocol did not abort.

The secrecy proof uses two main technical ingredients. The first is the Leftover Hash Lemma applied to a classical-quantum state $\rho_{A_1^n E}$, which converts a lower bound on the smooth min-entropy of Alice's raw key into a bound on the secrecy of the extracted key. The second is the Generalized Entropy Accumulation Theorem (GEAT), which is used to establish such a lower bound on the smooth min-entropy in terms of the per round conditional von Neumann entropy. We state these two theorems as follows.

\begin{theorem}[Leftover Hash Lemma on $\rho_{A_1^n E}$ \cite{Renner05}] \label{thm:leftoverhash}
Let $\rho_{A_1^n E}$ be the classical-quantum state of Alice's raw key string $A_1^n \in \{0,1\}^n$ and the adversary's quantum side information $E$ immediately before privacy amplification. Let $\Omega$ denote the event that Protocol 1 does not abort. Let $\mathcal{F}$ be a two-universal family of hash functions $F:\{0,1\}^n \to \{0,1\}^{\ell}$. Let $F$ be sampled independently and publicly. The final key is defined as $K := F(A_1^n)$. Then, for every $\eps \geq 0$, 
\begin{equation}
    \frac12 \left\|\rho_{K E \mid \Omega} - \tau_K \otimes \rho_{E \mid \Omega} \right\|_1 \leq \eps + \frac12 \cdot 2^{ -\frac12 \left( H_{\min}^{\eps} \bigl(A_1^n \mid X=0,E\bigr)_{\rho \mid \Omega} - \ell \right) }, 
\end{equation} 
where $\tau_K := 2^{-\ell}\mathbb{I}_K$ denotes the maximally mixed state on $\{0,1\}^{\ell}$.
\end{theorem}

\begin{theorem}[Generalized Entropy Accumulation Theorem (GEAT) \cite{metger2022generalised}]\label{thm:geat} Consider quantum channels $\mathcal{M}_i : R_{i-1}E_{i-1} \to A_i R_i E_i \,, i \in \{1,\dots,n\},$ satisfying the following non-signaling condition: for each channel $\mathcal{M}_i$, there exists a quantum channel $\mathcal{R}_i : E_{i-1} \to E_i$ such that
\begin{equation}\label{eq:geat_nonsignalling} 
    \operatorname{Tr}_{A_i R_i} \circ \mathcal{M}_i = \mathcal{R}_i \circ \operatorname{Tr}_{R_{i-1}} . 
\end{equation} 
Then, for any initial state $\rho^0_{R_0E_0}$, the smooth min-entropy of the outputs $A^n := A_1\cdots A_n$ conditioned on the final side information $E_n$ satisfies 
\begin{align}\label{eq:geat_min_entropy_bound} 
H_{\min}^{\varepsilon} \!\left( A^n \mid E_n \right)_{ \mathcal{M}_n \circ \cdots \circ \mathcal{M}_1 \left(\rho^0_{R_0E_0}\right) } &\geq \sum_{i=1}^{n} \inf_{\omega} H\!\left( A_i \mid E_i \widetilde{E}_{i-1} \right)_{ \mathcal{M}_i(\omega) } - O(\sqrt{n}) \,, \\
&\ge nh - c_0 - c_1\sqrt{n}  \,,
\end{align} 
where $\widetilde{E}_{i-1} \equiv R_{i-1}E_{i-1}$ is a purifying system for the input to $\mathcal{M}_i$, and the infimum is taken over all states $\omega$ on the systems $R_{i-1}E_{i-1}\widetilde{E}_{i-1}$. Appendix \ref{app:geat} contains the specific values of the constants $c_0$ and $c_1$, which depend on an affine min-tradeoff function $f$ as well.
\end{theorem}

We first define the channels in Protocol \hyperref[pr:one]{1} which satisfy the non-signaling condition required by GEAT, and give a short proof of this claim.

\begin{lemma}[Non-signaling condition]\label{lemma:non_signaling}
    Consider the following CPTP map in a general round $i$: $\{\mathcal{M}_i: R_{i-1}E_{i-1} \rightarrow K_{A,i}K_{B,i,j}C_iR_iE_i\}$ of Protocol 1, where $K_{A,i}$ is the set of Alice's two key bits in the $i-$th round, $R=R_AR_{B_1}...R_{B_{N-1}}$ are the internal registers of Alice and Bobs' devices, $C_i$ is the classical information computed by Alice and Bob (e.g., for privacy amplification), and $E_i=X^i_1Y^{i,1}_1...Y^{i,N-1}_1T^i_1Q_E$ is Eve's side information for that round. The state $\omega_{RE_{i-1}}^{i-1}$ is given as input to the map $\mathcal{M}_i$. Then, there exist corresponding reduced channels $\{\mathcal{R}_i : E_{i-1} \to E_i\}$ such that $\operatorname{Tr}_{A_i R_i} \circ \mathcal{M}_i = \mathcal{R}_i \circ \operatorname{Tr}_{R_{i-1}}$.
\end{lemma}
\begin{proof}
Consider the channel $\mathcal{R}_i : E_{i-1} \to E_i$ constructed as follows, following the steps of protocol. So, given the register $E_{i-1}$, we generate the register $E_i$ by sampling $T_i,X_i,Y_i^1,...,Y_i^{N-1}$ as Alice and Bobs do in Protocol 1 at the start of the round, and then simulating Eve's operations on her part of the state. Clearly, the distribution on these samples is identical to what is produced by $\mathcal{R}$. Since we assumed the devices of Alice and Bob cannot communicate with Eve, the simulated output marginal of $\mathcal{M}_i$ will be independent of $R_{i-1}$. Thus, if we trace out the $R_{i-1}$ register, the action of $\mathcal{R}_i$ is identical to the expected action when the channel $\mathcal{M}_i$ is applied to the state $\omega_{RE_{i-1}}^{i-1}$. Thus, $\operatorname{Tr}_{A_i R_i} \circ \mathcal{M}_i = \mathcal{R}_i \circ \operatorname{Tr}_{R_{i-1}}$. 
\end{proof}

In the absence of a closed-form analytical \emph{min-tradeoff} function for even the 2-player Magic Square Game, we resort to obtaining a lower bound on the per round conditional von Neumann entropy $H(A|X=0,E)$ of Protocol 1 via numerical methods. We use the techniques of \cite{brown2024device} to formulate a constrained minimization problem, which can be solved using the NPA hierarchy method \cite{navascues2007bounding}. We first reduce the entropy lower bound of $\pMSG_N$ to $\pMSG_2=\MSG$ to simplify our problem as follows, and then solve for the case of $\MSG$ as a semi-definite program. 

Since this is perhaps the \emph{key component} of our security analysis, we state and prove the following lemma in full generality for Parity-$G$ games as defined in Section \ref{sec:parityG}; setting $G=\MSG$ would then follow as a special case.

\begin{lemma}[Entropy Lower Bound Reduction from the $N$-player game $\pG{G}_N$ to the two-player game $G$]
\label{lem:pmsg-Nto2-reduction}
Let $N\geq2$ and $x^\ast\in\cX$, where the latter is Alice's input set with $x^\ast$ being her key input. Let $A$ denote Alice's $k$-bit
output, and $E$ denote Eve's quantum system. Let us assume there exists a convex function $f_{G}:[0,1]\to\mathbb R$ such that the per round conditional von Neumann entropy is lower bounded as
\begin{equation}
 H(A\mid X=x^\ast, E)_{G} \geq f_{G}(\omega)
 \label{eq:msg-entropy-tradeoff}
\end{equation}
over all possible quantum strategies by Eve, for the ordinary two-player $G$,
where $\omega$ is its winning probability for $G$ and $E$ is arbitrary quantum side information. Then,
\begin{equation}
 H(A\mid X=x^\ast,E)_{\pG{G}_N} \geq f_{G}(\omega') \,,
 \label{eq:pmsg-from-msg-entropy-tradeoff}
\end{equation}
where $\omega'$ is the winning probability for $\pG{G}_N$.
\end{lemma}

\begin{proof}
Consider a strategy $\mathcal S$ for $\pG{G}_N$. Let the output distribution of the parity-Bobs $B_2...B_{N-1}$ be described by the random variable $C$. Since the parity-Bobs receive no input, and by the non-signaling assumption, their output distribution $q_c \coloneq \Pr[C=c]$ is independent of Alice's and $B_1$'s inputs. Indeed, if $M_{a|x}$, $N_{b|y}$, and $D_{c}$ are the corresponding measurement operators, then
\begin{align}
 \Pr[C = c\mid X=x,Y=y]
 &=\sum_{a,b}\tr\!\left[\rho\,
     M_{a|x}\otimes N_{b|y}\otimes D_{c}\right]\\
 &=\tr\!\left[\rho\,
     \id\otimes\id\otimes D_{c}\right] =q_{c}. \label{eq:pmsg-transcript-independent}
\end{align}

For any output $c$ sampled from $C$, let $p=p(c)$ be the corresponding parity string output. Conditioned on this, Alice and Bob $B_1$ reduce to playing $G$, by the proof of Theorem~\ref{thm:parityG-values}. In particular, Alice and $B_1$ can play a strategy for $G$ with Alice's relabeled POVMs $M^{'x}_{a}=M^{x}_{a\oplus\sgn(x,p(c))}$. Hence the conditional, relabeled strategy $\mathcal S \mid p(c)$ is equivalent to a $G$ strategy. Let $\omega_{c}$ denote its winning probability. By \eqref{eq:pmsg-transcript-independent} and the definition of independent events, the conditional distribution of inputs $(X,Y)$ over all $c$ also remains the same. Therefore,
\begin{equation}
 \omega'_{\pG{G}_N}(\mathcal S) = \sum_c q_c\omega_c \,,
 \label{eq:pmsg-score-mixture}
\end{equation}
where $\omega'_{\pG{G}_N}(\mathcal S)$ is the winning probability of strategy $\mathcal S$ for $\pG{G}_N$. 

Now, by Lemma~\ref{lem:mask-props} (ii), $M'^{x^*}=M^{x^*}$. By definition of $f_G$, this implies $H(A\mid E,X=x^\ast,C=c) \ge f_G(\omega_c)$. So, applying \eqref{eq:msg-entropy-tradeoff} in each branch gives the following lower bound for the strategy $\mathcal S$, splitting $E = \Tilde{E}C$,
\begin{align}
 H(A\mid E, X=x^\ast) &=\sum_{c}q_{c} H(A\mid \Tilde{E},X=x^\ast,C=c)\\
 &\geq\sum_{c}q_{c}
   f_{G}(\omega_{c})\\
 &\geq f_{G}\!\left(
   \sum_{c}q_{c}\omega_{c}\right) \label{eq:jensen} \\
 &=f_{G}\!\left(\omega'_{\pG{G}_N}(\mathcal S)\right),
\end{align}
where \eqref{eq:jensen} uses Jensen's inequality and the last line follows from \eqref{eq:pmsg-score-mixture}. Since we assumed $\mathcal S$ was arbitrarily chosen, the lemma follows.
\end{proof}

\begin{lemma}[Conditional von Neumann entropy lower bound]\label{lemma:cvn-entropy}
    The conditional von Neumann entropy of Alice's classical key bitstring $H(A|X=0,Q_E) \ge \inf\limits_{\omega}  H\!\left( A_i \mid E_i \widetilde{E}_{i-1} \right)_{ \mathcal{M}_i(\omega) } \ge f(\omega) = 33.4\omega -31.47$, imposing the winning probability threshold $\omega \ge \omega_{\mathrm{tol}}$.
\end{lemma}
\begin{proof}
Following the formalism in \cite{brown2024device}, we can express a lower bound as follows.
\begin{lemma}[\cite{brown2024device}]
    Let $m\in\mathbb N$, and let $t_1,\dots,t_m$ and $w_1,\dots,w_m$ be the nodes and weights of an $m$-point Gauss--Radau quadrature on $[0,1]$ with endpoint $t_m=1$. Let $\rho_{Q_AQ_BQ_E}$ be the initial quantum state shared between Alice, Bob, and Eve's devices, and let $\{M_{a|x^\ast}\}_a$ denote the measurement operators performed by Alice's device in response to the input $X=x^\ast$. Furthermore, for $i=1,\dots,m-1$, let $\alpha_i:=\frac{3}{2}\max\left\{\frac{1}{t_i},\frac{1}{1-t_i}\right\}$. Then $H(A|X=x^\ast,Q_E)$ is never smaller than \begin{equation} \label{eq:gauss-radau-entropy-bound} \begin{aligned} c_m + \sum_{i=1}^{m-1} \frac{w_i}{t_i\ln 2} \sum_a \inf_{\substack{ Z_a\in\mathcal B(Q_E)\\ \|Z_a\|\le \alpha_i }} \operatorname{Tr}\!\Big[ \rho_{Q_AQ_E} \Big( & M_{a|x^\ast}\otimes \big( Z_a+Z_a^\ast+(1-t_i)Z_a^\ast Z_a \big) \\ & + t_i \big( I_{Q_A}\otimes Z_aZ_a^\ast \big) \Big) \Big], \end{aligned} \end{equation} where $c_m=\sum_{i=1}^{m-1}\frac{w_i}{t_i\ln 2}$. Moreover, these lower bounds converge to $H(A|X=0,Q_E)$ as $m\to\infty$.
\end{lemma}
We also require the following additional constraint which describes the winning condition for the magic square with the winning probability threshold $\omega_\mathrm{tol}$.
\begin{align}
    \omega(\rho_{Q_AQ_BQ_E},M_{a|x^\ast},N_{b|y^\ast}) &\ge \omega_\mathrm{tol} \,,
\end{align}
Formulating and solving the resultant SDP via the NPA hierarchy techniques \cite{navascues2007bounding}, we get a lower bound function from a feasible solution to the dual SDP as prescribed by \cite{brown2024device}, as follows.
\begin{equation}
  f(\omega) = 33.4\omega -31.47 \,.
\end{equation}
This solution comes from solving the SDP for the $\MSG$ game\footnote{all code used for solving SDPs, and plotting key rates based on this in Section \ref{sec:key-rates}, can be found at \cite{MagicSquareCVN}}, but by Lemma $\ref{lem:pmsg-Nto2-reduction}$, we know this is also a valid lower bound for the $\pMSG_N$ game, for any $N\ge2$, since $f$ is clearly a convex -- indeed affine -- function.

\end{proof}

\noindent We now combine these results and are ready to prove our main theorem. 

\begin{theorem}[Secrecy of Protocol 1]\label{thm:secrecy}
    Conditioned on not aborting, for every $\eps>0$, Protocol 1 produces $\ell$ bits of secret key and is $\eps_{\mathrm{sec}}$-secret with
    \begin{equation}
        \eps_{\mathrm{sec}} \le \eps + \frac12 \cdot 2^{ -\frac12 \left(nh - c_0 - c_1\sqrt{n} - \mathrm{leak}_\eps - \ell \right) } 
    \end{equation}
    where $h=\min_{c^n \in \Omega} f\!\left(\operatorname{freq}(c^n)\right)$ corresponding to the affine min-tradeoff function $f(\omega) = 33.4\omega -31.47$, $\gamma$ is the probability of a test round, and the constants $c_0,c_1$ are as defined in Theorem \ref{thm:geat}.
\end{theorem}
\begin{proof}
        By Theorem \ref{thm:geat} and Lemma \ref{lemma:cvn-entropy}, since we satisfy Lemma \ref{lemma:non_signaling}, we can lower bound the smooth min-entropy of Alice's key bitstring $A^n$ before error correction as
    \begin{equation}
        H^{\eps}_{\min}(A^n|E_n) \ge nh - c_0 - c_1\sqrt{n} \,.
    \end{equation}
    Let the additional leakage in the broadcast, error correction and privacy amplification steps be denoted by $\mathrm{leak}_\eps$. The exact value of this depends on implementation details and will be negligible asymptotically. Applying Theorem \ref{thm:leftoverhash}, with this lower bound on $H^{\eps}_{\min}(A^n|E_n)$ and leakage proves our desired result.
\end{proof}

\begin{corollary}[Security of Protocol 1]\label{thm:security}
    Conditioned on not aborting, for every $\eps'>0$, Protocol 1 produces $\ell$ bits of secret key and is $\eps$-secure with 
    \begin{equation}
        \eps \le (N-1)(\eps_{EC}+\eps'_{EC}) + \left(1-\gamma\left(1-e^{-2\delta^2}\right)\right)^n + \eps' + \frac12 \cdot 2^{ -\frac12 \left(n \cdot f(\omega) - O(\sqrt{n}) - \mathrm{leak}_\eps  - \ell \right)} 
    \end{equation}
\end{corollary}
\begin{proof}
    This follows from Definition \ref{def:cka}; and the statements of Theorem \ref{thm:correctness} and Theorem \ref{thm:secrecy}.
\end{proof}

%% file: noisy-key-rates.tex
\section{Asymptotic key rate analysis with depolarizing noise}\label{sec:key-rates}
The winning probability threshold we choose in an implementation of $\pMSG_N$ should inherently be a function of the noise we expect in the channels. If we expect our channels to be perfectly noiseless, then we should set our winning probability threshold to be perfectly $\omega(v)=1$. In this case, losing any single game would mean that we detect an eavesdropper Eve, and abort the protocol. This implies that in the asymptotic regime, we get a guarantee that our key is perfectly secure with high probability, by monogamy of entanglement and self-testing property of the magic square game. Now, we define the key rate of a protocol as
\begin{equation}
    r \coloneq \frac{\ell}{n} = (1-\gamma)[H(A |X=0,T=1,E) - H(A|X=0,T=1,B)]
\end{equation}
where $\ell$ is the number of bits of the final secret key produced by the protocol, $n$ is the number of rounds, $\gamma$ is the test round probability, $H(A|X=0,T=1,E)$ is the per round conditional von Neumann entropy of Alice's raw key in a test round given by $T_i=1$ (so Alice's input $X=0$), conditioned on Eve's side information, and $H(A|X=0,T=1,B)$ is the conditional entropy of Alice's final key (after error correction and privacy amplification) with respect to the final key of any Bob $B_i$. For this calculation, we assume that all the outputs generated in the test rounds are published publicly and used for parameter estimation, so the players get secret bits of key as a fraction of the key rounds only.

For calculating key rates, we consider an \emph{honest} implementation of the protocol, which means that the states prepared and measurements performed by the devices are honest, \emph{up to some noise threshold}. This means that the adversary Eve can adversarially prepare the state, and try to guess the key, but any measurement outcomes the dishonest devices report should be indistinguishable from the honest measurement outcome distribution on a noisy state. This is an upper bound on the number of bits of key we can get in a fully adversarial setting with noisy \emph{and} dishonest measurements, and it is a simplified scenario which makes it amenable to comparing different protocols. 

In particular, assuming absence of noise, setting $\omega_{\mathrm{tol}}=1$, $\gamma=0$, and having shared states which are honestly prepared, the key rate for $\pMSG_N$ is 2. This is maximal since at most 2 bits of key can be produced per round. More generally, any quantum key distribution or conference key agreement protocol, whose security relies on the monogamy of entanglement, achieves maximal key rate in the absence of noise. 

But since, realistically, our channels are not noiseless, we must set a winning probability threshold which is below 1. Unfortunately, this also results in key rates which are lower. However, we can still get good lower bounds on the key rate as a function of the noise, by lower bounding $\omega_{\mathrm{tol}}$ as a function of the noise and applying Lemma \ref{lemma:cvn-entropy} to get the corresponding number of bits per round. We first consider a global depolarizing noise channel when preparing the state, which is a standard assumption for a noise model in our setting, and also subsumes noisy measurements made by the players. We note that our techniques extend to any general noise model beyond depolarizing. The state $\rho$ that all the players will share now is 
\begin{equation}
    \rho = (1-q)\ket{\psi}\bra{\psi} + q\frac{\id_{4^N}}{4^N} \,,
\end{equation}
where $q \in [0,1]$ is the noise parameter and $\ket{\psi}=\ket{GHZ_N}^{\otimes 2}$.  Note that the quantum bit error rate (QBER) $Q=\Pr( A\ne B_i)$, which is the probability that Alice and some Bob $B_i$'s bits disagree on average over their 2 bits, is $Q=q/2$. Now, the expected probability $p_{\mathrm{exp}}$ of them winning $\pMSG_N$ is $1-q/2$, since with probability $1-q$, the state is correct and with probability $q$, both the outputs of Alice and Bob $B_1$, and the parity bits, are all random. So, we set $\omega_{\mathrm{tol}} = p_{\mathrm{exp}} = 1-q/2=1-Q$ and compute the \emph{asymptotic} key rate $r_\infty$ of $\pMSG_N$ by defining $r_\infty \coloneq \lim_{n\rightarrow\infty} \frac{l}{n}$. So,
\begin{equation}\label{eq:msg-global-lower-bound}
    r_{\infty}(\pMSG_N) = (1-\gamma)[f(p_{\mathrm{exp}}) - h(Q)] = (1-\gamma)[f(1-Q) - h(Q)]
\end{equation}
where $f$ was defined in Lemma \ref{lemma:cvn-entropy}, and $h$ is the binary entropy function. We will also set $\gamma = O\left(\frac{\log(n)}{n}\right)$, so $\lim_{n\rightarrow\infty} \gamma = 0$. 

\paragraph{Local Depolarizing Noise Channels}
We may also consider a local depolarizing noise model where each \emph{pair of qubits} that any Bob receives \emph{independently} becomes maximally mixed with probability $q$. This model is relevant in the setting where we assume the state is locally prepared by Alice's device, but the major source of noise comes from Alice sending out a pair of qubits from $\ket{GHZ_N}^{\otimes 2}$ to each of the Bobs, so we assume the quantum channel between Alice and any Bob $B_i$ is a depolarizing channel $\mathcal{D}:\mathcal{L}(\mathcal{H})\rightarrow\mathcal{L}(\mathcal{H})$, with $\dim(\mathcal{H})=4$, such that
\begin{equation}
    \mathcal{D}(\rho) = (1-q)\rho + q\frac{\id_4}{4}
\end{equation}
where $q \in [0,1]$. The QBER in this case becomes $Q=q/2$, similarly as before. However, the expected winning probability $p_{\mathrm{exp}}$ changes. If Bob $B_1$'s subsystem becomes maximally mixed, then the probability of winning $\pMSG_N$ becomes $1/2$, but if $B_1$ is noiseless and at least one parity-Bob qubit becomes maximally mixed, there is a $2/3$ chance of winning (if Alice picks the first row as input); in all other cases it is $1/2$ . So, 
\begin{align}
    p_{\mathrm{exp}} &= \Pr[\text{all $B_i$s ideal}]\cdot 1 + \Pr[B_1\text{ noisy}] \cdot \frac12 + \Pr[B_1\text{        ideal }, \exists i: \text{$B_i \ne B_1$ noisy}]\cdot \frac{2}{3}  \\ 
        &= (1-q)^{N-1} + \frac{q}{2} + \frac{2}{3}(1-q)\left(1-(1-q)^{N-2}\right) \\
        &= \frac{1}{2} + \frac{(1-q) + 2(1-q)^{N-1}}{6} \,.
\end{align}

\paragraph{Depolarizing Noise per qubit}
We can further consider a more fine grained noise model where \emph{each qubit} that any Bob receives becomes maximally mixed with probability $q \in [0,1]$. The QBER remains $Q=q/2$, and the expected winning probability $p_{\mathrm{exp}}$ is again calculated through a case analysis as follows.
\begin{align}
    p_{\mathrm{exp}} &= (1-q)^{2N-2} + 2q(1-q)\left[\frac{5}{9}+\frac{(1-q)^{N-2}}{18}\right] + \frac{q^2}{2} + (1-q)^2\biggl[2(1-(1-q)^{N-2})\\&\nonumber\quad (1-q)^{N-2} \cdot \frac{7}{9} + (1-(1-q)^{N-2})^2 \cdot \frac{2}{3} \biggl] \\
        &= \frac{1}{2} + \frac{(1-q)}{9} + \frac{(1-q)^2}{18} + \frac{1}{9}(1-q)^{N-1} + \frac{1}{9}(1-q)^N + \frac{1}{9}(1-q)^{2N-2} \label{eq:pmsg_noise_local}
\end{align}

\subsection{Comparison to the Parity-CHSH game}
We compare the asymptotic key rate of our CKA protocol using $\pMSG_N$ to its precursor, the CKA protocol using Parity-CHSH game as defined in \cite{ribeiro2018fully}, which we denote by $\pCHSH_N$. It is analogously defined, with Alice and Bob $B_1$ playing CHSH games in the test rounds instead of the magic square game. This game only requires 1 GHZ state to be shared between the players per round instead of 2, but it produces at most 1 bit of key per round for this reason as well. We also calculate the key rate for $\pCHSH$ as follows, in the presence of global depolarizing noise. Note that there is no dependence on the number of players $N$, since we are considering global noise on the entire state, regardless of the Hilbert space dimension.
\begin{equation}
    r_{\infty}(\pCHSH_N) = f'\left(\frac{1}{2}+\frac{1-2Q}{2\sqrt{2}}\right) - h(Q)
\end{equation}
where $f'(\omega) = 1-h\left(\frac{1}{2}+\frac{1}{2}\sqrt{(4\omega-2)^2-1}\right)$ asymptotically. We again set $\gamma = O\left(\frac{\log(n)}{n}\right)$, so $\lim_{n\rightarrow\infty} \gamma = 0$, and this term vanishes from the $r_{\infty}$ expression. We explain this comparison in Figure \ref{fig:entropy-global-noise}.

\begin{figure}[h!]
    \centering
    \includegraphics[width=\linewidth]{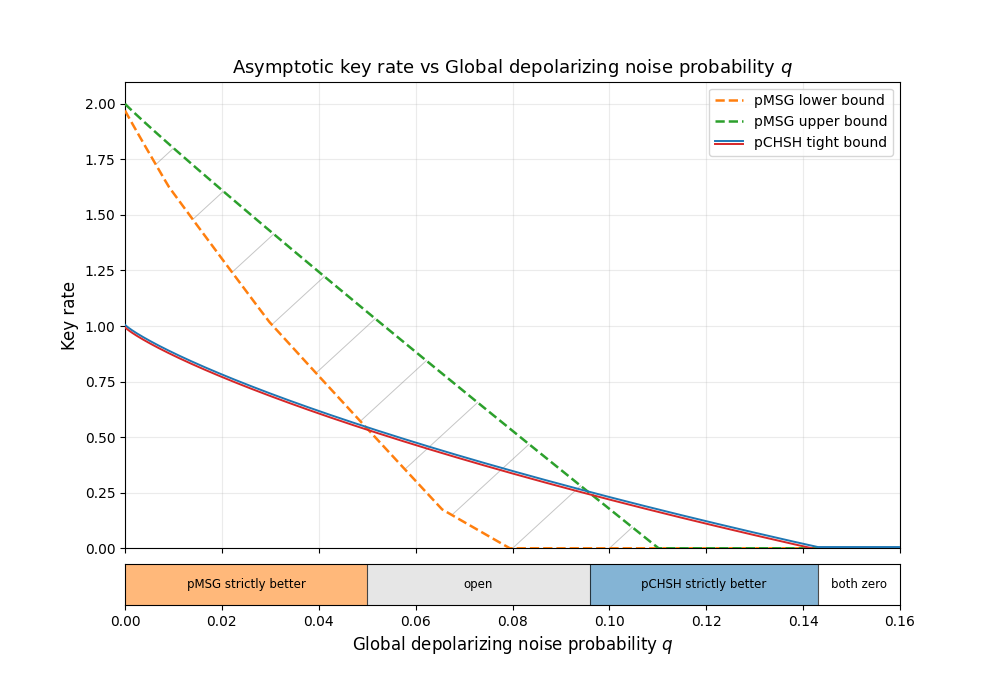}
    \caption{Key rate $r=H(A|X=0,T=1,E) - h(q/2)$ as a function of the \emph{global} depolarizing noise probability $q$. The per round entropy lower bound for $\pMSG$ is computed from Eq. (\ref{eq:msg-global-lower-bound}) and the techniques of Lemma~\ref{lemma:cvn-entropy}. We use a piecewise linear lower bound instead of just a single affine function since we just need a per round bound. The upper bound for $\pMSG$ is calculated from the strategy described in the paragraph below. We refer to prior work for a tight entropy upper and lower bound for $\pCHSH$ \cite{grasselli2023boosting}. The \emph{true} key rate for $\pMSG$ lies in the shaded region. The bars below the graph show which game achieves higher key rate in some noise regime; the middle section is open since our bounds for $\pMSG$ are not tight.}
    \label{fig:entropy-global-noise}
\end{figure}

\paragraph{The pMSG upper bound} To prove an upper bound on the per round conditional von Neumann entropy of $\pMSG$, it suffices to show that a strategy for Eve exists such that given a particular expected winning probability $\omega$, $H(A|X=0,T=1,E) \le g(\omega)$ for some function $g$. Of course, if $\omega \le 8/9$, which is the classical winning probability for $\pMSG$, then $g(\omega)=0$ since a classical strategy would succeed up to this probability. Similarly, if $\omega=1$ which is the quantum winning probability for $\pMSG$, then $g(\omega)=2$, by monogamy of entanglement. Now surprisingly, it turns out that there exists a (quantum) strategy for Eve such that $g(\omega) = 0$ for $\omega \lesssim 0.935$ \cite{adamson2020quantum}. Since this seems to not have been discussed in the DI-QKD literature to the best of the authors' knowledge, we describe this strategy in Appendix \ref{app:msg_upper}. Eve then employs a convex-combination attack \cite{lukanowski2023upper} (between this strategy and the honest strategy), which gives us a reasonably good upper bound $g(\omega)$ for any $\omega$.

Now, we step back and discuss our model of comparison. For the purpose of comparing these two games, if we allow Eve to adversarially prepare states such that they maximize her guessing probability for the bits of secret key, then global depolarizing noise is the correct model to consider, since we put no bound on the dimension of the state Alice's device prepares, so a per qubit noise model is not well-defined. 

However, note that, with an \emph{honest} strategy, the maximal key rate for $\pMSG_N$ is 2, and for $\pCHSH_N$ it is 1. This is mainly due to the fact that, in $\pMSG_N$, we allow for the players to share two GHZ states in parallel per round for the honest strategy. This suggests we can assume a noise model which fixes the dimension of the states, and allows for a fairer comparison by adding noise proportional to the number of qubits the players share, so arbitrarily increasing the size of this state doesn't improve key rate for free.

\paragraph{Local Depolarizing Noise Channels} 
Similarly as in the case of analyzing $\pMSG_N$ with noisy channels, the QBER remains $Q=q/2$. Following the same case analysis as previously, the noisy winning probability is as follows. Note that this noise model is the same as the depolarizing noise per qubit case, since Alice sends  exactly one qubit to every Bob for $\pCHSH_N$, unlike in $\pMSG_N$ where she had to send 2 qubits each.
\begin{align}
    p_{\mathrm{exp}} &= (1-q)^{N-1} \cdot \left(\frac{1}{2}+\frac{1}{2\sqrt{2}}\right) + \frac{q}{2}+ (1-                          q)\left(1-(1-q)^{N-2}\right) \cdot \left(\frac{1}{4}+ \frac{1}{2}\cdot                           \left(\frac{1}{2}+\frac{1}{2\sqrt{2}}\right)\right) \\
                     &=  \frac{1}{2} + \frac{(1-q) + (1-q)^{N-1}}{4\sqrt{2}} \,. \label{eq:pchsh_noise_local}
\end{align}

\begin{figure}[h!]
    \centering
    \includegraphics[width=0.9\linewidth]{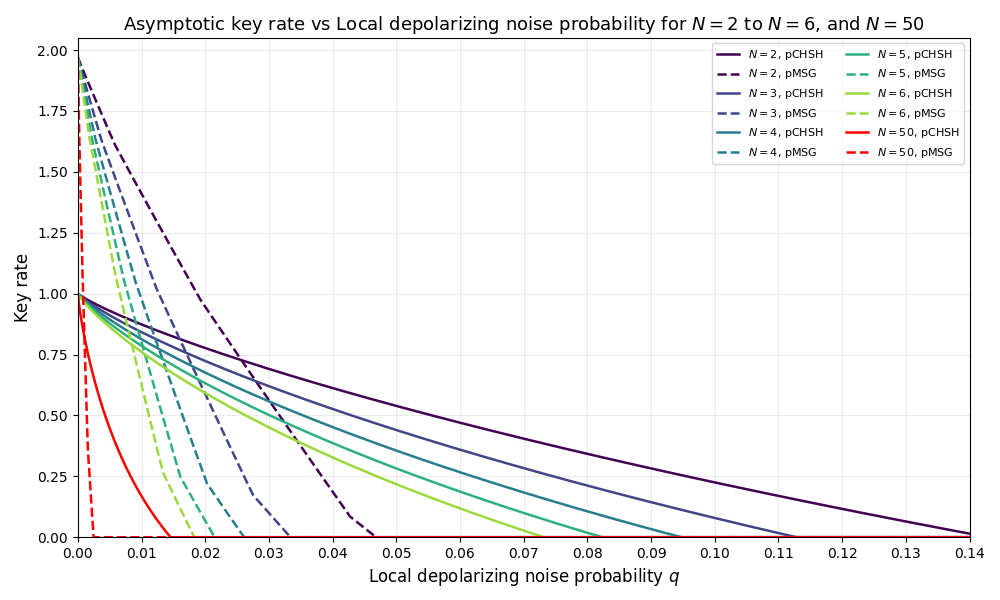}
    \caption{Key rate $r=H(A |X=0,T=1,E) - h(q/2)$ as a function of the \emph{local, per qubit,} depolarizing noise probability $q$, for various values of the number of parties $N$. The per round entropy lower bound for $\pMSG$ is computed from Eq. (\ref{eq:pmsg_noise_local}) and the techniques of Lemma~\ref{lemma:cvn-entropy}. We use a piecewise linear lower bound instead of just a single affine function since we just need a per round bound. We refer to Eq. (\ref{eq:pchsh_noise_local}) and  prior work for the (tight) lower bound for $\pCHSH$ \cite{grasselli2023boosting}.}
    \label{fig:entropy-varyingNsq}
\end{figure}

\paragraph{Comparison with Parallel CHSH.}
\phantomsection
\label{par:parallel}
Perhaps the most natural DI-CKA proposal is to use $N-1$ pairwise DI-QKDs to share keys between a designated leader and each other party, allowing the leader to one-time pad their choice of key to all others. If each DI-QKD was based on a two-player non-local game, the resulting DI-CKA protocol would either be based on \textit{sequential} two-player games, reducing the bits of key gained per round by a factor of $N-1$, or would be based on \textit{parallel} two-player games, which are actually new hard-to-analyze $N$-party games where the leader gets $N-1$ inputs while all others get a single input.

In \cite{ribeiro2018fully}, the main point of comparison for their CKA protocol based on the $\pCHSH$ game is $N-1$ sequential DI-QKD protocols, each based on the CHSH game. Since we show that $\pMSG$ obtains improved key rates over $\pCHSH$ in the low noise regime, this implies we also improve over $N-1$ sequential DI-QKD CHSH protocols in this regime. 

We may also attempt to compare with $N-1$ DI-QKD protocols played in \emph{parallel}, but that would require us to compute entropy lower bounds for a new $N$-party parallel CHSH game. This is essentially unfeasible with current techniques such as \cite{brown2024device} since the size of the resulting SDPs would scale with $N$. The only reason we are able to analyze the security of $\pMSG$ using an SDP is because the security reduces to the 2-party case as shown in Lemma \ref{lem:pmsg-Nto2-reduction}. We analyze a possibly tractable variant which we term the \emph{Parallel Parity-CHSH game} in Appendix \ref{app:par-pchsh}, which may serve as a proxy for comparison with $N-1$ parallel DI-QKDs. We leave the full comparison with Parallel CHSH, or indeed any parallel two-party game including the Parallel Magic Square game, to future work.

%% file: acknowledgements.tex
\section{Acknowledgments}
This work was completed while all authors were employed by Visa Inc. The authors thank Adithya Bhat and Jonathan Lu for helpful conversations and the late Dr. Anderson C. A. Nascimento for his kindness and inspiring curiosity. AS is supported by a TQI Graduate Fellowship. RR is supported by an NSF Graduate Research Fellowship.

\paragraph{Statement of AI Use.} Generative AI was used to write the code to solve semi-definite programs (found here \cite{MagicSquareCVN}), based on the implementation by \cite{zhen2023device}, which gave us the bounds for Lemma \ref{lemma:cvn-entropy}. It also helped write the code for calculating the key rates in Figure \ref{fig:entropy-global-noise} and Figure \ref{fig:entropy-varyingNsq} as well. It was also used to aid in proofreading and some parts of writing the paper. It did not contribute to any of the main ideas or conceptual insights of this work.

%% file: appendix.tex
\section*{Supplementary material}

\input{impossibility}

\section{The Generalized Entropy Accumulation Theorem} \label{app:geat}
We restate the full theorem here for reference, including the specific values of the constants $c_0,c_1$ mentioned in the main text.
\begin{theorem}[Generalized Entropy Accumulation Theorem (GEAT) \cite{metger2022generalised}] Consider quantum channels $\mathcal{M}_i : R_{i-1}E_{i-1} \to A_i R_i E_i \,, i \in \{1,\dots,n\},$ satisfying the following non-signaling condition: for each channel $\mathcal{M}_i$, there exists a quantum channel $\mathcal{R}_i : E_{i-1} \to E_i$ such that
\begin{equation}\label{eq:geat_nonsignalling_appendix} 
    \operatorname{Tr}_{A_i R_i} \circ \mathcal{M}_i = \mathcal{R}_i \circ \operatorname{Tr}_{R_{i-1}} . 
\end{equation} 
Then, for any initial state $\rho^0_{R_0E_0}$, the smooth min-entropy of the outputs $A^n := A_1\cdots A_n$ conditioned on the final side information $E_n$ satisfies 
\begin{align}\label{eq:geat_min_entropy_bound_appendix} 
H_{\min}^{\varepsilon} \!\left( A^n \mid E_n \right)_{ \mathcal{M}_n \circ \cdots \circ \mathcal{M}_1 \left(\rho^0_{R_0E_0}\right) } &\geq \sum_{i=1}^{n} \inf_{\omega} H\!\left( A_i \mid E_i \widetilde{E}_{i-1} \right)_{ \mathcal{M}_i(\omega) } - O(\sqrt{n}) \,, \\
&\ge nh - c_0 - c_1\sqrt{n}  \,,
\end{align} 
where $\widetilde{E}_{i-1} \equiv R_{i-1}E_{i-1}$ is a purifying system for the input to $\mathcal{M}_i$, and the infimum is taken over all states $\omega$ on the systems $R_{i-1}E_{i-1}\widetilde{E}_{i-1}$. The constants $c_0$ and $c_1$ are as follows, corresponding to an affine min-tradeoff function $f$, with $h=\min_{c^n \in \Omega} f\!\left(\operatorname{freq}(c^n)\right)$.
\[
c_0 \coloneq
\frac{(2-\eta)\eta^2\log(1/\Pr_{\rho_n}[\Omega]+\eta^2)g(\varepsilon)}{3(\ln 2)^2V^2(2\eta-1)^3}2^{\frac{1-\eta}{\eta}\left(2\log d_A+\operatorname{Max}(f)-\operatorname{Min}_{\Sigma}(f)\right)}\ln^3\left(2^{2\log d_A+\operatorname{Max}(f)-\operatorname{Min}_{\Sigma}(f)}+e^2\right),
\]
\begin{flalign*}
c_1 \coloneq \sqrt{\frac{2\ln(2)V^2}{\eta}
\left(g(\varepsilon)+(2-\eta)\log(1/\Pr_{\rho_n}[\Omega])\right)}\,,
&&
\end{flalign*}
such that
\[
\eta \coloneq \frac{2\ln(2)}{1+2\ln(2)}, \qquad
g(\varepsilon) \coloneq -\log\left(1-\sqrt{1-\varepsilon^2}\right), \qquad
V \coloneq \log\left(2d_A^2+1\right)+\sqrt{2+\operatorname{Var}(f)}\,.
\]
\end{theorem}

\section{The Magic Square game strategy} \label{app:msg_upper}
Since the security of $\pMSG$ reduces to the security of the Magic Square game, as shown in Lemma \ref{lem:pmsg-Nto2-reduction}, we simply analyze the latter. Our goal is to come up with a strategy for Eve in the Magic square game DI-QKD protocol, such that the entropy of Alice's key bits is zero while obtaining a winning probability $\omega > 8/9$, which is strictly greater than the best classical winning probability. We follow the idea introduced in \cite{adamson2020quantum}, in the context of a DI-QKD protocol. 

We will design an attack from the perspective of the adversary Eve who prepares the devices that Alice and Bob use, and then calculate what winning probability we obtain. Recall that Alice and Bob both sample their inputs from the uniform distribution over $\{0,1,2\}$. When either device gets an input 0, this corresponds to measuring in the Z basis, which includes all key rounds and some test rounds. So the devices deterministically generate the same but known output bits (say 10) on such rounds. So, Eve learns all the key bits. Conversely, when either device gets either 1 or 2 as an input, this must be a test round. This new game now reduces to a \emph{CHSH game} as shown in \cite{adamson2020quantum}, where the input to a device is either $0$ or from the set \{1,2\}. The devices play the optimal CHSH strategy on the latter inputs, winning with probability $\cos^2(\pi/8)$. Now, we can calculate the overall winning probability of this strategy for the Magic square game, with known key rounds inputs, as $5/9 \cdot 1 + 4/9 \cdot \cos^2(\pi/8) \approx 0.935$. Perhaps surprisingly, this turns out to be greater than $8/9$, concluding our argument. The optimality of this winning probability remains open, but an SDP calculation shows us that it cannot possibly be more than $\approx 0.945$.

\section{The Parallel Parity-CHSH game} \label{app:par-pchsh}
Like $\pCHSH$, the parallel Parity-CHSH game is also an $N$-player game, but where all players receive 2 sets of inputs, each input set corresponding to a $\pCHSH$ game. They win the parallel game iff they win both $\pCHSH$ games. A similar argument to Lemma \ref{lem:pmsg-Nto2-reduction} would then show that this reduces to winning two CHSH games in parallel. We can design a conference key agreement protocol based on this game, similarly as the $\pCHSH$ is used for \cite{ribeiro2018fully}.

Now, we could try comparing the key rates of the conference key protocol based on $\pMSG$ to the parallel Parity-CHSH game, between Alice and each of the Bobs. The honest strategy in both cases now is to share 2 GHZ states. But the optimal classical winning probability of both CHSH games in parallel is $5/8 = 0.625$ which is greater than $(3/4)^2=9/16=0.5625$ as we might expect. The optimal quantum winning probability now becomes $\cos^4(\pi/8) \approx 0.7285$ \cite{ambainis2026optimal}. It is unclear how this would change the key rates since we would need to re-derive bounds on $H(A|E)$ and key rates for this new game. In particular, we could again try applying the method by \cite{brown2024device} to get lower bounds on the per round conditional von Neumann entropy of Parallel Parity-CHSH, which allows us to lower bound the key rate.

We note that we expect its \emph{finite} round key rate (as opposed to asymptotic key rate which is what we consider in the main text for simplicity) to be worse than that of $\pMSG$ when the key round probability is a \emph{constant}, particularly in the low noise regime, due to the following fact. The conditional von Neumann entropy $H(A|E)$ is calculated conditioned on Alice receiving a certain input, as described in Lemma \ref{lem:pmsg-Nto2-reduction}. For the Parallel Parity-CHSH game, it would be conditional on Alice receiving the bits 00 as input, which occurs with probability $1/4$. So, when the noise is zero, only $1/4$'th of the key rounds will result in a valid key bit. Let there be a $\mu=1-\gamma$ fraction of key rounds over all rounds. Thus, $H(A|E) = p(X=00) \cdot H(A|E,X=00) =  \mu \cdot 1/4 \cdot 2 = \mu/2$ when the noise probability $q=0$, and it will further decrease as $q$ increases. The corresponding calculation for $\pMSG$ gives $H(A|E) = \mu \cdot 1/3 \cdot 2 = 2\mu/3 > \mu/2$. For $\pCHSH$, it is $\mu \cdot 1/2 \cdot 1 = \mu/2$. So, in the finite round, low noise regime, it suggests $\pMSG$ obtains better key rates over both these games. The asymptotic case would require a separate analysis. 

This also gives us some intuition why increasing the number of qubits shared between all parties per round in a conference key agreement doesn't arbitrarily let us improve key rate, since the probability of a valid key round would naively go down. For $M$ parallel CHSH games, we would get $M$ bits of key on valid key rounds when the input to (say) Alice is $0^M$, but the probability of this is $2^{-M}$, which goes down exponentially as $M$ increases. This suggests we may also need better formulations for asymptotic key rate calculations, since various current definitions would only analyze the key rate on valid key rounds when Alice's input is $0^M$, which is clearly not indicative when running this experiment in a practical sense.

We leave a tight comparison between the key rates obtained from $\pMSG, \pCHSH$ and the Parallel Parity-CHSH game, as an interesting open question for future work.

%% file: impossibility.tex
\section{Impossibility of a Multiparty Consensus Magic Square}\label{sec:obstruction}

In this section, we prove why a natural multiparty magic square game fails to be useful for a conference key agreement protocol, which motivates us to define $\pMSG_N$ in the framework of Parity-$G$ games.

\begin{definition}[Literal multiparty magic square $\MSG_N$]\label{def:msgn}
$N$ parties: Alice and $B_1,\dots,B_{N-1}$. Alice receives a row $x\in\{0,1,2\}$ and outputs
an even triple $a$; each $B_j$ receives the \emph{same} column $y\in\{0,1,2\}$ (uniform) and
outputs an odd triple $b^{(j)}$. The players wins iff $a_y=b^{(j)}_x$ for all $j$ --- Alice
reaches \emph{cell-consensus} with every Bob. (This is a consensus predicate: winning forces
$a_y=b^{(1)}_x=\cdots=b^{(N-1)}_x$.)
\end{definition}

\begin{theorem}[Monogamy obstruction]\label{thm:obstruction}
$\wq(\MSG_N)<1$ for every $N\ge3$. Hence $\MSG_N$ is not a pseudo-telepathy game.
\end{theorem}

\begin{proof}
Suppose a strategy wins $\MSG_N$ with probability $1$. Purify everything to a register $\ket\psi$ on
$\cH_A\otimes\cH_{B_1}\otimes\cdots\otimes\cH_{B_{N-1}}\otimes\cH_\Eve$.

\emph{Step 1 (bipartite reduction).} Fix any $B_j$. Marginalizing to $(A,B_j)$: Alice's input
$x$ and $B_j$'s input $y$ are uniform and independent, Alice outputs an even triple, $B_j$ an
odd triple, and the $j$-th winning clause $a_y=b^{(j)}_x$ holds with certainty for every
$(x,y)$. Thus $(A,B_j)$ wins the bipartite magic square (Definition~\ref{def:msg}) with
probability $1$.

\emph{Step 2 (self-test pins a single core).} Alice performs \emph{one} fixed measurement (her row observables), regardless of which Bob is considered. By the self-testing property of the Magic Square game \cite{coladangelo2017robust}, it is identical up to isometry in every application. Hence there is a single two-qubit core $\Acore\subseteq\cH_A$ such that,
for \emph{every} $j$, perfect $(A,B_j)$ play forces $\Acore$ to be maximally entangled with
$B_j$'s core. (Alice cannot split into independent per-Bob cores: the Peres observables
generate the full operator algebra on her two-qubit core, which acts irreducibly on
$\mathbb C^4$, so a single observable cannot act non-trivially on several independent cores.)

\emph{Step 3 (monogamy).} Apply Step 2 with $j=2$: $\rho_{\Acore,B_{2}}$ is pure maximally
entangled. By Lemma~\ref{lem:monogamy}, $\Acore$ is then in a product state with $B_{3}$, so
the entanglement between $\Acore$ and $B_{3}$ is zero. But Step 2 with $j=3$ requires $\Acore$
maximally entangled with $B_{3}$ (two ebits), which is a contradiction. Therefore no perfect
strategy exists: $\wq(\MSG_N)<1$.
\end{proof}

\begin{remark}\label{cor:nocons}
The above theorem shows that the obvious idea for constructing a genuinely multipartite \emph{consensus} pseudo-telepathy game of the magic-square
form, for $N\ge3$, does not work. A DI conference key therefore cannot be extracted from this game, motivating the two-basis design of Section~\ref{sec:protocol}: a non-consensus certification test plus a separate consensus key basis. Our Parity-$G$ games framework explicitly requires this in the form of test rounds and key rounds.

We leave it as an intriguing open problem to construct a genuinely multipartite \emph{consensus} pseudo-telepathy game, from which we could extract a DI conference key with a single measurement setting. 
\end{remark}

%% file: disclaimer.tex
\medskip

\noindent \textbf{Disclaimer}

\medskip

\noindent Case studies, comparisons, statistics, research and recommendations are provided ``AS IS'' and intended for informational purposes only and should not be relied upon for operational, marketing, legal, technical, tax, financial or other advice.  Visa Inc. neither makes any warranty or representation as to the completeness or accuracy of the information within this document, nor assumes any liability or responsibility that may result from reliance on such information.  The Information contained herein is not intended as investment or legal advice, and readers are encouraged to seek the advice of a competent professional where such advice is required.
 
These materials and best practice recommendations are provided for informational purposes only and should not be relied upon for marketing, legal, regulatory or other advice. Recommended marketing materials should be independently evaluated in light of your specific business needs and any applicable laws and regulations. Visa is not responsible for your use of the marketing materials, best practice recommendations, or other information, including errors of any kind, contained in this document.